\documentclass[12pt]{article} 
\usepackage[utf8]{inputenc} 

\usepackage{graphicx} 
\usepackage{amsmath, amsfonts, amssymb, color, amsxtra}
\usepackage{bm}
\usepackage{enumerate}
\usepackage{lineno}
\usepackage[utf8]{inputenc}
\usepackage{verbatim}
\usepackage{appendix}
\usepackage{amsmath}
\usepackage{amsfonts}
\usepackage{amssymb}
\usepackage{amsthm}
\usepackage{graphicx}
\usepackage{natbib}
\usepackage{color}
\usepackage{caption}
\usepackage{multirow}
\usepackage[affil-it]{authblk}
\usepackage{algpseudocode, algorithm, algorithmicx}
\usepackage{xr}

\usepackage[top=1in, bottom=1in, left=0.8in, right=0.8in]{geometry} 
\usepackage{graphicx} 
\usepackage{amsmath, amsfonts, color}
\usepackage{bm}
\usepackage{enumerate}

\usepackage{booktabs} 
\usepackage{array} 
\usepackage{paralist} 
\usepackage{verbatim} 
\usepackage{subfig} 

\usepackage{fancyhdr} 
\usepackage{sectsty}
\allsectionsfont{\sffamily\mdseries\upshape} 

\usepackage[nottoc,notlof,notlot]{tocbibind} 
\usepackage[titles,subfigure]{tocloft} 

\usepackage{setspace}
\newtheorem{theorem}{Theorem}[section]
\newtheorem{lemma}[theorem]{Lemma}
\newtheorem{proposition}[theorem]{Proposition}

\newcommand{\bx}{{\bm x}}
\newcommand{\bb}{{\bm b}}
\newcommand{\bX}{{\bm X}}

\newcommand{\by}{{\bm y}}

\newcommand{\be}{{\bm e}}

\newcommand{\ba}{{\bm a}}

\newcommand{\bbeta}{{\bm \beta}}

\newcommand{\btau}{{\bm \tau}}
\newcommand{\bSigma}{{\bm \Sigma}}

\newcommand{\btheta}{{\bm \theta}}

\newcommand{\bepsilon}{{\bm \epsilon}}
\newcommand{\beps}{{\bm \epsilon}}

\newcommand{\bzero}{{\bm 0}}

\newcommand{\bI}{{\bm I}}
\newcommand{\bA}{{\bm A}}

\newcommand{\br}{{\bm r}}
\newcommand{\bQ}{{\bm Q}}

\newcommand{\bP}{{\bm P}}

\newcommand{\cA}{\mathcal{A}}

\newcommand{\cD}{\mathcal{D}}
\newcommand{\cI}{\mathcal{I}}

\newcommand{\cE}{\mathcal{E}}
\newcommand{\cV}{\mathcal{V}}
\newcommand{\cS}{\mathcal{S}}
\newcommand{\cM}{\mathcal{M}}
\newcommand{\cB}{\mathcal{B}}

\newcommand{\E}{\mathbb{E}}
\newcommand{\Var}{\mathrm{Var}}
\newcommand{\Cov}{\mathrm{Cov}}
\newcommand{\Cor}{\mathrm{Cor}}
\newcommand{\CI}{\mathrm{CI}}
\DeclareMathOperator*{\argmin}{arg\,min}

\newcommand{\A}{\mathcal{A}}
\newcommand{\supp}{\mathrm{supp}}
\DeclareMathOperator{\sign}{sign}
\providecommand{\keywords}[1]{\textbf{\textit{Keywords---}} #1}
\newcommand{\RN}[1]{%
  \textup{\uppercase\expandafter{\romannumeral#1}}%
}
\newcommand{\trace}{\mathrm{tr}}

\newcommand\smallO{
  \mathchoice
    {{\scriptstyle\mathcal{O}}}
    {{\scriptstyle\mathcal{O}}}
    {{\scriptscriptstyle\mathcal{O}}}
    {\scalebox{.7}{$\scriptscriptstyle\mathcal{O}$}}
  }

\date{} 

\author{Sen Zhao\thanks{sen-zhao@sen-zhao.com.}} 
\author{Daniela Witten\thanks{dwitten@u.washington.edu.}}
\author{Ali Shojaie\thanks{ashojaie@u.washington.edu.}}
\affil{University of Washington}

\title{In Defense of the Indefensible: \\ A Very Na\"{i}ve Approach to High-Dimensional Inference}

\begin{document}

\maketitle
\def\spacingset#1{\renewcommand{\baselinestretch}%
{#1}\normalsize} \spacingset{1}
\abstract{
A great deal of interest has recently focused on 
conducting inference on the parameters in a high-dimensional linear model. 
 In this paper, we consider a simple and very na\"{i}ve two-step procedure for this task,  in which we (i) fit a lasso model in order to obtain a subset of the variables, and (ii) fit a least squares model on the lasso-selected set. Conventional statistical wisdom tells us that we cannot make use of the standard statistical inference tools for the resulting least squares model (such as confidence intervals and $p$-values), since we peeked at the data twice: once in running the lasso, and again in fitting the least squares model.  However, in this paper, we show that under a certain set of assumptions, with high probability, the set of variables selected by the lasso is identical to the one selected by the noiseless lasso and is hence deterministic. Consequently, the na\"{i}ve two-step approach can yield asymptotically valid inference. We utilize this finding to develop the \emph{na\"ive confidence interval}, which can be used to draw inference on the regression coefficients of the model selected by the lasso, as well as the \emph{na\"ive score test}, which can be used to test the hypotheses regarding the full-model regression coefficients.
}

\keywords{Confidence interval; Lasso; $p$-value; Post-selection inference; Significance testing.}

\spacingset{1.5}

\section{Introduction}\label{sec:intro}

In this paper, we consider the linear model
\begin{align}\label{eq:linmodel}
\by&=\bX\bbeta^\ast+\bepsilon,
\end{align}
where $\bX = [\bx_1,\dots,\bx_p]$ is an $n\times p$ deterministic design matrix, $\bepsilon$ is a vector of independent and identically distributed errors with $\E\left[\epsilon_i\right]=0$ and $\Var\left[\epsilon_i\right]=\sigma^2_\bepsilon$, 
and $\bbeta^\ast$ is a $p$-vector of coefficients.  
Without loss of generality, we assume that the columns of $\bX$ are centered and standardized, such that  $\sum_{i=1}^nX_{(i, k)} = 0$ and $\|\bx_k\|_2^2=n$ for $k=1,\ldots,p$. 

When the number of variables $p$ is much smaller than the sample size $n$, estimation and inference for the vector $\bbeta^\ast$ are straightforward. 
For instance, estimation can be performed using ordinary least squares, and inference can be conducted using classical approaches \citep[see, e.g.,][]{GelmanHillregression, WeisbergOLS}.

As the scope and scale of data collection have increased across virtually all fields, there is an increase in data sets that are \emph{high dimensional}, in the sense that the number of variables, $p$, is larger than the number of observations, $n$. In this setting, classical approaches for estimation and inference of $\bbeta^\ast$ cannot be directly applied. 
    In the past 20 years, a vast statistical literature has focused on estimating  $\bbeta^\ast$ in high dimensions. In particular,  penalized regression methods, such as the lasso \citep{tibs1996},
\begin{equation}\label{eq:lasso}
\hat \bbeta_\lambda = \argmin_{\bb \in \mathbb{R}^p}\left\{\frac{1}{2n}\|\by-\bX \bb\|_2^2 + \lambda \|\bb\|_1 \right\},
\end{equation}
can be used to estimate $\bbeta^\ast$. 
 However, the topic of inference in the high-dimensional setting remains relatively less explored, despite promising recent work in this area. 
  Roughly speaking, recent work on inference in the high-dimensional setting falls into two classes: (i) methods that examine the null hypothesis $H_{0,j}^\ast: \beta^\ast_j=0$; and (ii) methods that make inference based on a sub-model. We will review these two classes of methods in turn.
 
 First, we review methods that examine the null hypothesis $H_{0,j}^\ast: \beta^\ast_j=0$, i.e. that the variable $\bx_j$ is unassociated with the outcome $\by$, conditional on \emph{all other variables}. It might be tempting to estimate $\bbeta^\ast$ using the lasso \eqref{eq:lasso}, and then (for instance) to construct a confidence interval around $\hat\beta_{\lambda,j}$. Unfortunately, such an approach is problematic, because $\hat\bbeta_\lambda$ is a biased estimate of $\bbeta^\ast$. To remedy this problem, we can apply a one-step adjustment to $\hat\bbeta_\lambda$, such that under appropriate assumptions, the resulting \emph{debiased estimator} is asymptotically unbiased for $\bbeta^\ast$. This is similar to the idea of two-step estimation in nonparametric and semiparametric inference \citep[see, e.g., ][]{Hahn1998,Hiranoetal2003}.
 With the one-step adjustment, $p$-values and confidence intervals can be constructed around this debiased estimator. Such an approach is taken by the  low dimensional projection estimator \citep[LDPE;][]{ZhangZhang2014LDPE, vandeGeeretal2014LDPE},  the debiased lasso test with unknown population covariance \citep[SSLasso;][]{JavanmardMontanari2013GLMSSLasso, javanmard2013confidence}, the debiased lasso test with known population covariance \citep[SDL;][]{JavanmardMontanari2014SDLTheory}, 
 and the decorrelated score test \citep[dScore;][]{NingLiu2015decor}. See \citet{Dezeureetal2015hdi} for a review of such procedures.
 In what follows, we will refer to these and related approaches for testing $H_{0,j}^\ast: \beta^\ast_j=0$ as \emph{debiased lasso  tests}.

Next, we review recent work that makes statistical inference based on a sub-model.
 Recall that the challenge in high dimensions stems from the fact that when $p > n$, classical statistical methods cannot be applied; for instance, we cannot even perform ordinary least squares (OLS). This suggests a simple approach: given an index set $\cM\subseteq\{1,\dots,p\}$, let $\bX_{\cM}$ denote the columns of $\bX$ indexed by $\cM$. Then, we can consider performing inference \emph{based on the sub-model} composed only of the features in the index set $\cM$. That is, rather than considering the model \eqref{eq:linmodel}, we consider the sub-model
 \begin{align}\label{eq:submodel}
\by&=\bX_\cM \bbeta^{(\cM)}+\bepsilon^{(\cM)}.
\end{align}
In~\eqref{eq:submodel}, the notation $\bbeta^{(\cM)}$ and $\bepsilon^{(\cM)}$ emphasizes that the true regression coefficients and corresponding noise are functions of the set $\cM$. 

Now, provided that $| \cM | < n$, we can perform estimation and inference on the vector $\bbeta^{(\cM)}$ using classical statistical approaches. 
For instance, we can consider building confidence intervals $\CI_j^{(\cM)}$ such that for any $j\in\cM$,
\begin{align}
\label{eq:goalofposi}
\Pr\left[\beta_j^{(\cM)}\in \CI_j^{(\cM)}\right]\geq 1-\alpha.
\end{align}
At first blush, the problems associated with high dimensionality have been solved!
  
Of course, there are some problems with the aforementioned approach. The first problem is that the  coefficients 
 in the sub-model \eqref{eq:submodel} typically are not the same as the coefficients in the original model \eqref{eq:linmodel} \citep{berk2013valid}. 
 Roughly speaking, the problem is that the coefficients in the model \eqref{eq:linmodel} quantify the linear association between a given variable and the response, \emph{conditional on the other $p-1$ variables}, whereas the coefficients in the sub-model \eqref{eq:submodel} quantify the linear association between a variable and the response, \emph{conditional on the other $|\cM|-1$ variables in the sub-model}.  
The true regression coefficients in the sub-model are of the form
\begin{align}
\bbeta^{(\cM)}\equiv\left(\bX_{\cM}^\top\bX_{\cM}\right)^{-1}\bX_{\cM}^\top\bX\bbeta^\ast.
\label{eq:betaM}
\end{align}
 Thus, $\bbeta^{(\cM)} \neq \bbeta^\ast_{\cM}$ unless $\bX^\top_{\cM}\bX_{\cM^c}\bbeta^\ast_{\cM^c}=\bzero$. 
  To see this more concretely, consider the following example with $p=4$ deterministic variables. Let 
\[
\frac{1}{n}\bX^\top\bX=  \begin{bmatrix}
    1 & 0 & 0.6 & 0 \\
    0 & 1 & 0.6 & 0 \\
    0.6 & 0.6 & 1 & 0 \\
    0 & 0 & 0 & 1
  \end{bmatrix}, 
\]
and set $\bbeta^\ast=(1, 1, 0, 0)^\top$. The above design matrix does not satisfy the strong irrepresentable condition needed for selection consistency of lasso \citep{ZhaoYu2006}. Thus, if we take $\cM$ to equal the support of the lasso estimate, i.e.,
\begin{equation}
\cM= \hat \A_\lambda \equiv \supp\left(\hat\bbeta_\lambda\right) \equiv \left\{ j: \hat\beta_{\lambda,j} \neq 0  \right\},
\label{Alambda}
\end{equation}
then it is easy to verify that for some $\lambda$, $\cM=\{2,3\}$, in which case
\begin{align*}
\bbeta^{(\cM)}	&=\left[\frac{1}{n}\bX^\top\bX\right]_{(\{2, 3\}, \{2, 3\})}^{-1}\left[\frac{1}{n}\bX^\top\bX\right]_{(\{2, 3\}, \{1, 2, 3, 4\})}\bbeta^\ast \\
			&= \begin{bmatrix}
    1 & 0.6 \\
    0.6 & 1
  \end{bmatrix}^{-1}
  \begin{bmatrix}
    0 & 1 & 0.6 & 0 \\
    0.6 & 0.6 & 1 & 0
  \end{bmatrix}  
  \begin{bmatrix}
    1 & 1 & 0 & 0
  \end{bmatrix}^\top \\
  &=  \begin{bmatrix}
    0.4375 \\
    0.9375
  \end{bmatrix} \neq\begin{bmatrix}
    1 \\
    0
  \end{bmatrix}=\bbeta^\ast_{\cM}.
\end{align*}

The second problem that arises in restricting our attention to the sub-model \eqref{eq:submodel} is that in practice, the index set $\cM$ is not pre-specified. 
Instead, it is typically chosen based on the data. 
The problem is that if we construct the index set $\cM$ based on the data, and then apply classical inference approaches on the vector $\bbeta^{(\cM)}$, the resulting $p$-values and confidence intervals will not be valid \citep[see, e.g.,][]{Potscher1991PoSI, Kabaila1998PoSI, LeebPostcher2003PoSI, LeebPotscher2005PoSI, LeebPotscher2006PoSI, LeebPostcher2006perfPoSI, LeebPotscher2008PoSI, Kabail2009PoSI, berk2013valid}.  
This is because we peeked at the data twice:  once to determine which variables to include in $\cM$, and then again to test hypotheses associated with those variables.  
 Consequently, an extensive  recent body of literature has focused on the task of performing inference on  $\bbeta^{(\cM)}$ in \eqref{eq:submodel} given that $\cM$ was chosen based on the data. \citet{Cox1975} proposed the idea of sample-splitting to break up the dependence of variable selection and hypothesis testing. \citet{WassermanRoeder2009posi} studied sample-splitting in application to the lasso, marginal regression and forward step-wise regression. \citet{Meinshausenetal2009posi} extended the single-splitting proposal of \citet{WassermanRoeder2009posi} to multi-splitting, which improved statistical power and reduced the number of falsely selected variables. \citet{berk2013valid} instead considered simultaneous inference, which is universally valid under all possible model selection procedures without sample-splitting.  More recently, \citet{lee2015exact} and \citet{Tibshiranietal2016PoSI} studied the geometry of the lasso and sequential regression, respectively, and proposed exact post-selection inference methods conditional on the random set of selected variables. See \citet{TaylorTibshirani2015posi} for a review of post-selection inference procedures.

The procedures outlined above successfully address the second problem  arising from restricting the attention to a sub-model, namely the randomness of the set $\cM$. However, they do not address the first problem regarding the difference in the target of inference \eqref{eq:betaM} unless $\bX^\top_{\cM}\bX_{\cM^c}\bbeta^\ast_{\cM^c}=\bzero$. In the case of lasso, valid inference for $\bbeta^\ast_{\A^\ast}$ can be obtained if $\cM = \hat\A_\lambda = \A^\ast$. However, among other conditions, this requires the strong irrepresentable condition, which is known to be a restrictive condition that is not likely to hold in high dimensions \citep{ZhaoYu2006}.

In a recent \emph{Statistical Science} paper, \citet{LeebetalPoSI2015}  performed a simulation study, in which they obtained a set $\cM$ using variable selection, and then calculated \emph{``na\"{i}ve'' confidence intervals} for $\bbeta^{(\cM)}$ using ordinary least squares, \emph{without accounting for the fact that the set $\cM$ was chosen based on the data}. Of course, conventional wisdom dictates that the resulting confidence intervals will be much too narrow. In fact, this is what  \citet{LeebetalPoSI2015} found, when they used  best subset selection  to construct the set $\cM$.
  However, surprisingly, when the lasso was used to construct the set $\cM$, the confidence intervals induced by \eqref{eq:goalofposi} had approximately correct coverage.
  This is in stark contrast to the existing literature!

In this paper, we present a theoretical justification for the empirical finding in \citet{LeebetalPoSI2015}. The main idea of our paper is to establish selection consistency of the lasso estimate \emph{with respect to its noiseless counterpart}. This result allows us to perform valid inference for the support of the noiseless lasso without needing post-selection or sample splitting strategies. Furthermore, we use our theoretical findings to also develop the \emph{na\"ive score test}, a simple procedure for testing  the null hypothesis $H^\ast_{0,j}:\beta^\ast_j=0$ for all $j=1,\dots,p$.

The rest of this paper is organized as follows. In Sections~\ref{sec:nllasso} and \ref{sec:waldsims}, we focus on post-selection inference: we seek  to perform inference on $\bbeta^{(\cM)}$ in \eqref{eq:submodel}, where $\cM$ is selected based on the lasso, i.e., $\cM = \hat\A_\lambda$ \eqref{Alambda}.  In Section~\ref{sec:nllasso}, we point out a previously overlooked scenario in selection consistency theory: although $\hat\A_\lambda$ \eqref{Alambda} is random, with high probability it is equal to the support of the noiseless lasso under relatively mild regularity conditions. This result implies that we can use classical methods for inference on $\bbeta^{(\cM)}$, when $\cM = \hat\A_\lambda$.  In Section~\ref{sec:waldsims}, we provide empirical evidence in support of these theoretical findings.  In Sections~\ref{sec:waldscore} and Section~\ref{sec:numeric}, we instead focus on the task of performing inference on $\bbeta^\ast$ in \eqref{eq:linmodel}. We propose the na\"ive score test in Section~\ref{sec:waldscore}, and study its empirical performance in Section~\ref{sec:numeric}. We end with a discussion of future research directions in Section~\ref{sec:disc}. Technical proofs are relegated to the online Supplementary Materials.

We now introduce some notation that will be used throughout the paper. We use ``$\equiv$'' to denote equalities by definition, and ``$\asymp$'' for the asymptotic order. We use $1\{\cdot\}$ for the indicator function; ``$\vee$'' and ``$\wedge$'' denote the maximum and minimum of two real numbers, respectively. For any real number $a\in\mathbb{R}$, $a_+\equiv a\vee 0$. Given a set $\cS$, $|\cS|$ denotes its cardinality and $-\cS \equiv\cS^c$ denotes its complement. We use bold upper case fonts to denote matrices, bold lower case fonts for vectors, and normal fonts for scalars. We use symbols with a superscript ``$\ast$'', e.g., $\bbeta^\ast$ and $\A^\ast\equiv\supp(\bbeta^\ast)$, to denote the true population parameters associated with the full linear model \eqref{eq:linmodel};  
we use symbols superscripted by a set in the parentheses, e.g., $\bbeta^{(\cM)}$, to denote quantities related to the sub-model \eqref{eq:submodel}. Symbols subscripted by ``$\lambda$'' and with a hat, e.g., $\hat\bbeta_\lambda$ and $\hat\A_\lambda$, denote parameter estimates from the lasso estimator \eqref{eq:lasso} with tuning parameter $\lambda>0$; symbols subscripted by ``$\lambda$'' and without a hat, e.g., $\bbeta_\lambda$ and $\A_\lambda\equiv\supp(\bbeta_\lambda)$, are associated with the noiseless lasso estimator \citep{vandeGeerBuhlmann2009, vandeGeer2017}, 
\begin{align}\label{eq:nllasso}
\bbeta_\lambda\equiv\argmin_{\bb\in\mathbb{R}^p}\left\{\frac{1}{2n}\E\left[\left\|\by-\bX\bb\right\|_2^2\right]+\lambda\left\|\bb\right\|_1\right\}.
\end{align}
For any vector $\bb$, matrix $\bSigma$, and index sets $\cS_1$ and $\cS_2$, we use $\bb_{\cS_1}$ to denote the sub-vector of $\bb$ comprised of elements of $\cS_1$, and  $\bSigma_{(\cS_1, \cS_2)}$ to denote the sub-matrix of $\bSigma$ with rows in $\cS_1$  and columns in $\cS_2$. 

\section{Theoretical Justification for Na\"{i}ve Confidence Intervals}\label{sec:nllasso}

Recall that $\bbeta^{(\hat\A_\lambda)}$ was defined in \eqref{eq:betaM}. The simulation results of \citet{LeebetalPoSI2015} suggest that if we perform ordinary least squares using the variables contained in the support set of the lasso, 
$\hat\A_\lambda$,  then the classical confidence intervals associated with the least squares estimator,
\begin{align}\label{eq:posiols}
\tilde\bbeta^{(\hat\A_\lambda)}\equiv\left(\bX^\top_{\hat\A_\lambda}\bX_{\hat\A_\lambda}\right)^{-1}\bX^\top_{\hat\A_\lambda}\by,
\end{align}
have approximately  correct coverage, where correct coverage means that for all $j \in \hat\A_\lambda$, 
\begin{equation}
\Pr\left(\beta_{j}^{(\hat\A_\lambda)} \in \CI_j^{(\hat\A_\lambda)}\right) \geq 1-\alpha. 
\label{eq:ci}
\end{equation} 
We reiterate that in \eqref{eq:ci}, $\CI_j^{(\hat\A_\lambda)}$ is the confidence interval output by standard least squares software applied to the data $(\by, \bX_{\hat\A_\lambda})$. This goes against our statistical intuition: it seems that by fitting a lasso model and then performing least squares on the selected set, we are peeking at the data twice, and thus we would expect the confidence interval $\CI_j^{(\hat\A_\lambda)}$ to be much too narrow.

In this section, we present a theoretical result that suggests that, in fact, this ``double-peeking" might not be so bad. 
Our key insight is as follows: under certain assumptions, \emph{the set of variables selected by the lasso is deterministic and non-data-dependent with high probability}. Thus, fitting a least squares model on the variables selected by the lasso does not really constitute peeking at the data twice: effectively, with high probability, we are only peeking at it once. That means that the na\"{i}ve confidence intervals obtained from ordinary least squares will have approximately correct coverage, in the sense of \eqref{eq:ci}. 

We first  introduce the required conditions for our theoretical result. 
\begin{itemize}
\item[({\bf M1})] The design matrix $\bX$ is deterministic, with columns in \textit{general position} \citep{Rossetetal2004unique, Zhang2010MCP, Dossal2012unique, Tibshirani2013}. Columns of $\bX$ are standardized, i.e., for any $j=1,\dots,p$, $\bx_j^\top\bx_j = n$. The error $\bepsilon$ in \eqref{eq:linmodel} has independent entries and sub-Gaussian tails. The sample size $n$, dimension $p$ and tuning parameter $\lambda$ satisfy 
$$
\sqrt{\frac{\log(p)}{n}}\frac{1}{\lambda}\rightarrow 0.
$$
\item[({\bf E})]  Let $\hat\bSigma=\bX^\top\bX/n$. For any index set $\cI$ with $|\cI|=\mathcal{O}(q^\ast)$, let $\cB$ be any index set such that $\cB\supseteq\cI$, $|\cB\backslash\cI|\leq |\cI|$. Then for all $\ba\in\mathbb{R}^p$ that satisfy $\|\ba_{\cI^c}\|_1\leq \|\ba_{\cI}\|_1$, and $\|\ba_{\cB^c}\|_\infty\leq\min_{j\in(\cB\backslash\cI)}|a_j|$,
\begin{align*}
\phi^{\ast2}\equiv\liminf_{n\to\infty}\frac{\ba^\top\hat\bSigma\ba}{\|\ba_{\cB}\|_2^2}>0,
\end{align*}

In addition, the restricted sparse eigenvalue,
\begin{align*}
\phi^2(q)\equiv\sup_{\|\bb_{\A^{\ast c}}\|_0\leq q}\frac{\bb^\top\hat\bSigma\bb}{\|\bb\|_2^2},
\end{align*}
satisfies $\phi^2(q^\ast)=\mathcal{O}(\sqrt{\log(p)/q^\ast})$.  

\item[({\bf M2})] Recall that $\A^\ast\equiv\supp(\bbeta^\ast)$.  Let $\A_\lambda\equiv\supp(\bbeta_\lambda)$, $b_{\lambda\min}\equiv\min_{j\in\A_\lambda}|\beta_{\lambda,j}|$ with $\bbeta_\lambda$ defined in \eqref{eq:nllasso}. The signal strength in $\A_\lambda$ satisfies
\begin{align}
\lim_{n\to\infty}\frac{b_{\lambda\min}}{\lambda} =\xi > 0,
\end{align}
and the signal strength outside of $\A_\lambda$ satisfies
\begin{align}
\left\|\bX_{\A_\lambda^c}\bbeta^\ast_{\A_\lambda^c}\right\|_2=\mathcal{O}\left(\sqrt{\log(p)}\right).
\end{align}

\item[({\bf T})] The strong irrepresentable condition with respect to $\A_\lambda$ and $\bbeta_\lambda$ holds, i.e., there exists $\delta>0$ such that 
\begin{align}
\lim_{n\to\infty}\left\|\bX_{\A_\lambda^c}^\top\bX_{\A_\lambda}\left(\bX^\top_{\A_\lambda}\bX_{\A_\lambda}\right)^{-1}\sign\left(\bbeta_{\lambda,\A_\lambda}\right)\right\|_\infty < 1-\delta.
\end{align}
\end{itemize}

Condition ({\bf M1}) is mild and standard in literature. Note that in ({\bf M1}), we require the lasso tuning parameter $\lambda$ to approach zero at a slightly slower rate than $\sqrt{\log(p)/n}$ to control the randomness of the error $\bepsilon$. Most standard literature requires $\lambda>C\sqrt{\log(p)/n}$ for some constant $C>0$, which our condition also satisfies. 

In  ({\bf M2}), the requirement that $\lim_{n\to\infty}b_{\lambda,\min}/\lambda>0$ indicates that the noiseless lasso regression coefficients are either 0, or asymptotically no smaller than $\lambda$, which is larger than $\sqrt{\log(p)/n}$. Note that this assumption concerns variables chosen by the noiseless lasso and not the true model. 
In the second part of ({\bf M2}), $\big\|\bX_{\A_\lambda^c}\bbeta^\ast_{\A_\lambda^c}\big\|_2=\mathcal{O}(\sqrt{\log(p)})$ implies that the total signal strength of weak signal variables that are not selected by the noiseless lasso cannot be too large to be distinguishable from the sub-Gaussian noise. As shown in the proof of Proposition~\ref{THM:CONSISTENT}, the condition $b_{\lambda,\min}=\mathcal{O}(\lambda)$ is important in showing that $\lim_{n\to\infty}\Pr\left[\A_\lambda\subseteq\hat\A_\lambda\right]=1$, whereas $\big\|\bX_{\A_\lambda^c}\bbeta^\ast_{\A_\lambda^c}\big\|_2=\mathcal{O}\left(\sqrt{\log(p)}\right)$ is instrumental in showing $\lim_{n\to\infty}\Pr\left[\A_\lambda\supseteq\hat\A_\lambda\right]=1$.

Condition ({\bf E}) manifests the behavior of the eigenvalues of $\hat\bSigma$. Specifically, the first part of ({\bf E}) is the restricted eigenvalue condition \citep{bickeletal2009, vandeGeerBuhlmann2009}, except that instead of requiring $|\cI|=q^\ast$ as in \citet{bickeletal2009} and \citet{vandeGeerBuhlmann2009}, we here require $|\cI|=\mathcal{O}(q^\ast)$. The second part of ({\bf E}) is the sparse Reisz, or sparse eigenvalue condition \citep{ZhangHuang2008, BelloniChernozhukov2013qlambda}. Both the restricted eigenvalue and sparse Reisz conditions are standard and mild conditions in the literature. ({\bf E}) implies that $\log(p)/q^\ast\to \psi>1$.

Condition ({\bf T}) is the strong irrepresentable condition with respect to $\A_\lambda$ and $\bbeta_\lambda$. Condition ({\bf T}) is likely weaker than the classical irrepresentable condition with respect to $\bbeta^\ast$ and $\A^\ast$, proposed in \citet{ZhaoYu2006}, because the classical irrepresentable condition implies that $\A_\lambda=\A^\ast$ with large $n$, in which case ({\bf T}) holds as it becomes identical to the classical irrepresentable condition. Lemma~\ref{LEM:T} shows a sufficient condition for ({\bf T}), which is proven in Section \ref{sec:pfT} in the Supplementary Materials.

\begin{lemma}\label{LEM:T}
Suppose conditions  ({\bf M1}), ({\bf M2}) and ({\bf E}) hold. Condition ({\bf T}) holds if
\begin{align}
\lim_{n\to\infty}\frac{2\sqrt{2q^\ast}}{n\phi^{\ast 2}}\left\|\bX^\top_{\A_\lambda^c}\bX_{\A_\lambda}\right\|_\infty < 1-\delta. \label{eq:arend}
\end{align}
\end{lemma}
Condition \eqref{eq:arend} allows $\|\bX^\top_{\A_\lambda^c}\bX_{\A_\lambda}\|_\infty$ to diverge to infinity at a slower rate than $n/\sqrt{q^\ast}$. A more stringent version of condition \eqref{eq:arend} is presented in \citet{Voormanetal2015}.

We use ({\bf T}) to prove $\lim_{n\to\infty}\|\btau_{\lambda, \A_\lambda^c}\|_\infty \leq 1 - \delta$, where $\lambda n\btau_{\lambda}=\bX^\top\left(\bX\bbeta^\ast-\bX\bbeta_{\lambda}\right)$ is the stationary condition of \eqref{eq:nllasso}. Note that $\lim_{n\to\infty}\|\btau_{\lambda, \A_\lambda^c}\|_\infty \leq 1 - \delta$ is the required condition for Proposition~\ref{THM:CONSISTENT}, and ({\bf T}) might be sufficient but unncessary. However, ({\bf T}) is more standard and understandable than the condition $\lim_{n\to\infty}\|\btau_{\lambda, \A_\lambda^c}\|_\infty \leq 1 - \delta$.

We now present Proposition~\ref{THM:CONSISTENT}, which is proven in Section~\ref{sec:pfconsistent} of the online Supplementary Materials. 
  \begin{proposition}\label{THM:CONSISTENT}
Suppose conditions ({\bf M1}), ({\bf M2}), ({\bf E}) and ({\bf T}) hold. 
Then, we have $\lim_{n\to\infty}\Pr[\hat \A_\lambda= \A_\lambda]= 1$, where $\A_\lambda\equiv\supp(\bbeta_\lambda)$, with $\bbeta_\lambda$ defined in  \eqref{eq:nllasso}.
\end{proposition}

The proof of Proposition~\ref{THM:CONSISTENT} is in the same flavor as  \citet{MeinshausenBuhlmann2006, ZhaoYu2006, Tropp2006MI} and \citet{Wainwright2009PDW}, and is based on absorbing the contribution of the weak signals $\bX_{\A_\lambda^c}\bbeta_{\A_\lambda^c}^\ast$ into the noise vector. However, variable selection consistency asserts that $\Pr[\hat\A_\lambda=\A^\ast]\to1$, whereas Proposition~\ref{THM:CONSISTENT} states that $\Pr[\hat\A_\lambda=\A_\lambda]\to1$. Consequently, variable selection consistency requires the irrepresentable and signal strength conditions with respect to $\bbeta^\ast$, whereas Proposition~\ref{THM:CONSISTENT} requires the irrepresentable and signal strength conditions with respect to $\bbeta_\lambda$. These conditions are likely much milder than those for the variable selection consistency of the lasso, as confirmed in our simulations in Section~\ref{sec:waldsims}. In these simulations, we estimate $\Pr[\hat \A_\lambda= \A_\lambda]$ in 36 settings with two choices of $\lambda$. As shown in Table~\ref{tab:equalalambda}, in most settings, $\hat \A_\lambda= \A_\lambda$ with high probability, especially in cases with large $n$. We also estimated $\Pr[\hat \A_\lambda= \A^\ast]$ in the same 36 settings, with the same two choices of $\lambda$, as well as an additional choice of $\lambda$ that minimizes the cross-validated MSE; we found that $\Pr[\hat \A_\lambda= \A^\ast]=0$ in all $36 \times 3$ settings.

Based on Proposition~\ref{THM:CONSISTENT}, we could build asymptotically valid confidence intervals, as shown in Theorem \ref{THM:WALD2}.

Proposition~\ref{THM:CONSISTENT} 
suggests that asymptotically, we ``pay no price" for peeking at our data by performing the lasso: we should be able to perform downstream analyses on the subset of variables in $\hat\A_\lambda$ \emph{as though we had obtained that subset without looking at the data}. 
This intuition will be formalized in Theorem \ref{THM:WALD2}.

Theorem \ref{THM:WALD2}, which is proven in Section~\ref{sec:pfwald1} in the online Supplementary Materials, shows that $\tilde\bbeta^{(\hat\A_\lambda)}$ in \eqref{eq:posiols} is asymptotically normal, with  mean and variance suggested by classical least squares theory: that is, \emph{the fact that $\hat\A_\lambda$ was selected by peeking at the data has no effect on the asymptotic distribution of $\tilde\bbeta^{(\hat\A_\lambda)}$}. This result requires that $\lambda$ be chosen in a non-data-adaptive way. Otherwise, $\A_\lambda$ will be affected by the random error $\bepsilon$ through $\lambda$, which complicates the distribution of $\tilde\bbeta^{(\hat\A_\lambda)}$. 
 Theorem \ref{THM:WALD2} also requires Condition ({\bf W}), which is used to apply the Lindeberg-Feller Central Limit Theorem. This condition can be relaxed if the noise $\bepsilon$ is normally distributed.
\begin{itemize}
\item[({\bf W})] $\lambda$, $\bbeta^\ast$ and $\bX$ are such that $\lim_{n\to\infty}\| {\bf r}^w \|_\infty/ \|{\bf r}^w\|_2 \rightarrow 0$, where 
\[
{\bf r}^w \equiv  \be^j(\bX^\top_{\A_\lambda}\bX_{\A_\lambda})^{-1}\bX^\top_{\A_\lambda}, 
\]
and $\be^j$ is the row vector of length $|\A_\lambda|$ with the entry corresponding to $\beta^\ast_j$ equal to one, and zero otherwise.
\end{itemize}
\begin{theorem}\label{THM:WALD2}
Suppose $\lim_{n\to\infty}\Pr\left[\hat \A_\lambda= \A_\lambda\right]= 1$ and ({\bf W}) holds. Then, for any $j\in\hat\A_\lambda$,
\begin{align}\label{eq:Tw}
\frac{\tilde\beta_{j}^{(\hat\A_\lambda)}-\beta_j^{(\hat\A_\lambda)}}{\sigma_\bepsilon\sqrt{\left[(\bX^\top_{\hat\A_\lambda}\bX_{\hat\A_\lambda})^{-1}\right]_{(j,j)}}}\rightarrow_{d} \mathcal{N}\left(0, 1\right),
\end{align}
where $\tilde\bbeta^{(\hat\A_\lambda)}$ is defined in \eqref{eq:posiols} and $\bbeta^{(\hat\A_\lambda)}$ in \eqref{eq:betaM}, and ${\sigma_\bepsilon}$ is the variance of $\bepsilon$ in \eqref{eq:linmodel}.
\end{theorem}

The error standard deviation $\sigma_\bepsilon$ in \eqref{eq:Tw} is usually unknown. It can be estimated using various high-dimensional estimation methods, e.g., the scaled lasso \citep{SunZhang2012ScaledLasso}, cross-validation (CV) based methods \citep{Fanetal2012CVvar} or method-of-moments based methods \citep{Dicker2014var}; see a comparison study of high dimensional error variance estimation methods in \citet{Reidetal2016var}. Alternatively, Theorem~\ref{THM:ERRORVAR} shows that we could also consistently estimate the error variance using the post-selection OLS residual sum of square (RSS).

\begin{theorem}\label{THM:ERRORVAR}
Suppose $\lim_{n\to\infty}\Pr\left[\hat \A_\lambda= \A_\lambda\right]= 1$ and $\log(p)/(n-q_\lambda)\to0$, where $q_\lambda\equiv|\A_\lambda|$. Then
\begin{align}\label{eq:RSSestvar}
\frac{1}{n-\hat q_\lambda}\left\|\by-\bX_{\hat\A_\lambda}\tilde\bbeta^{(\hat\A_\lambda)}\right\|_2^2\to_p\sigma^2_\bepsilon,
\end{align}
where $\hat q_\lambda\equiv|\hat\A_\lambda|$.
\end{theorem}
Theorem~\ref{THM:ERRORVAR} is proven in Section~\ref{sec:pferrorvar} in the online Supplementary Materials. In \eqref{eq:RSSestvar}, $\by-\bX_{\hat\A_\lambda}\tilde\bbeta^{\hat\A_\lambda}$ is the fitted OLS residual on the sub-model \eqref{eq:submodel}. Also, $\log(p)/(n-q_\lambda)\to0$ is a weak condition: since $\log(p)/n\to0$, $\log(p)/(n-q_\lambda)\to0$ is satisfied if $\lim_{n\to\infty}q_\lambda/n<1$.

To summarize, in this section, we have provided a theoretical justification for a procedure that seems, intuitively, to be statistically unjustifiable:
  \begin{enumerate}
\item Perform the lasso  in order to obtain the support set $\hat\A_\lambda$; 
\item Use least squares to fit the sub-model containing just the features in $\hat\A_\lambda$;
\item Use the classical confidence intervals from that least squares model, without accounting for the fact that $\hat\A_\lambda$ was obtained by peeking at the data.
\end{enumerate}
 Theorem \ref{THM:WALD2} guarantees that  the na\"{i}ve confidence intervals in Step 3 will indeed have approximately correct coverage, in the sense of \eqref{eq:ci}.

\section{Numerical Examination of Na\"{i}ve Confidence Intervals}\label{sec:waldsims}

In this section, we perform simulation studies to examine the coverage probability \eqref{eq:ci} of the \emph{na\"{i}ve} confidence intervals obtained by applying standard least squares software to the data $\by, \bX_{\hat\A_\lambda}$.  
 
  Recall from Section~\ref{sec:intro} that \eqref{eq:ci} involves the probability that the confidence interval contains the quantity $\bbeta^{(\hat\A_\lambda)}$, which in general does not equal the population regression coefficient vector $\bbeta^\ast_{\hat\A_\lambda}$.  Inference for $\bbeta^\ast$ is discussed in Sections~\ref{sec:waldscore} and \ref{sec:numeric}.

The results in this section complement simulation findings in \citet{LeebetalPoSI2015}.

\subsection{Methods for Comparison}

Following Theorem~\ref{THM:WALD2}, for $\tilde\bbeta^{(\hat\A_\lambda)}$ defined in \eqref{eq:posiols}, and for each $j\in\hat\A_\lambda$, the 95\% na\"ive confidence interval takes the form
\begin{equation}
\small 
\CI_{j}^{(\hat\A_\lambda)}\equiv \left(\tilde\beta^{(\hat\A_\lambda)}_j-1.96\times\hat\sigma_\bepsilon\sqrt{\left[\bX^\top_{\hat\A_\lambda}\bX_{\hat\A_\lambda}\right]_{(j, j)}} \, , \,\,  \tilde\beta^{(\hat\A_\lambda)}_j+1.96\times\hat\sigma_\bepsilon\sqrt{\left[\bX^\top_{\hat\A_\lambda}\bX_{\hat\A_\lambda}\right]_{(j, j)}}\right).
\label{eq:naiveci}
\end{equation}
In order to obtain the set ${\hat\A_\lambda}$,  we must apply the lasso using some value of $\lambda$. By ({\bf M1}) we need $\lambda\succ\sqrt{\log(p)/n}$, which is slightly larger than the prediction optimal rate, $\lambda\asymp\sqrt{\log(p)/n}$ \citep{bickeletal2009, vandeGeerBuhlmann2009}. A data-driven way to obtain a larger tuning parameter it to use $\lambda_{\text{1SE}}$, which is the largest value of $\lambda$ for which the 10-fold CV prediction mean squared error (PMSE) is  within one standard error of the minimum CV PMSE \citep[see Section 7.10.1 in][]{Hastieetal2009}. However, the $\lambda$ selected by cross validation depends on $\by$ and may induce additional randomness in the set of selected coefficients, $\hat\A_\lambda$. This randomness can also impact the exact post-selection procedure of \citet{lee2015exact}. To address this issue, the authors proposed $\lambda_{\text{sup}}\equiv2\E[\|\bX^\top\be\|_\infty]/n$, where $\be\sim\mathcal{N}_n(\bzero, \hat\sigma_\bepsilon^2\bI)$. As an alternative to $\lambda_{\text{1SE}}$, we also evaluate $\lambda_{\text{sup}}$, where we simulate $\be$ and approximate the expectation based on the average of 1000 replicates. 
We compare the confidence intervals for $\bbeta^{(\hat\A_\lambda)}$ with those based on the exact lasso post-selection inference procedure of  \citet{lee2015exact}, which is implemented in the \textsc{R} package \texttt{selectiveInference}.  



In both approaches,  the standard deviation of errors, $\sigma_\bepsilon$ in \eqref{eq:linmodel}, is estimated either using the scaled lasso \citep{SunZhang2012ScaledLasso} or by applying Theorem~\ref{THM:ERRORVAR}. However, we do not examine the combination of $\lambda_{\text{sup}}$ with error variance estimated based on Theorem~\ref{THM:ERRORVAR}, because $\lambda_{\text{sup}}$ requires an estimate of the error standard deviation based on Theorem~\ref{THM:ERRORVAR}, whereas Theorem~\ref{THM:ERRORVAR} requires a suitable choice of $\lambda$ to be valid.

\subsection{Simulation Set-Up} \label{sec:simsetup-Sec3}
For the simulations, we consider two partial correlation settings for $\bX$, generated based on (i) a scale-free graph and (ii) a stochastic block model \citep[see, e.g.,][]{Kolaczyk2009networks}, each containing $p=100$ nodes. These  settings are relaxations of the simple orthogonal and block-diagonal settings, and are displayed in Figure~\ref{fig:graphsettings}.

In the scale-free graph setting, we used the \texttt{igraph} package in \textsc{R} to simulate an undirected, scale-free network $\mathcal{G}=(\cV,\cE)$ with power-law exponent parameter $\gamma=5$, and edge density 0.05.
 Here, $\cV=\{1,\ldots,p\}$ is the set of nodes in the graph, and $\cE$ is the set of edges. This resulted in a total of  $|\cE|=247$ edges in the graph. We then order the indices of the nodes in the graph so that the first, second, third, fourth, and fifth nodes correspond to the 10th, 20th, 30th, 40th, and 50th least-connected nodes in the graph.

In the stochastic block model setting, we first generate two dense Erd\H{o}s-R\'{e}nyi graphs \citep{ErdosRenyi1959ER, Gilbert1959ER} with five and 95 nodes, respectively. In each graph, the edge density is 0.3. We then add edges randomly between these two graphs to achieve an inter-graph edge density of 0.05. The indices of the nodes are ordered so that the nodes in the five-node graph precede the remaining nodes.

Next, for both graph settings, we define the  weighted adjacency matrix, $\bA$, as follows:
\begin{equation}
A_{(j,k)}=  \begin{cases} 1 & \mbox { for } j=k\\
 \rho &  \mbox{ for }(j,k)\in\cE\\
 0 & \mbox{ otherwise} \end{cases},
 \label{eq:A}
 \end{equation}
where 
 $\rho\in\{0.2,0.6\}$. 
We then set $\bSigma=\bA^{-1}$, and standardize $\bSigma$ so that $\Sigma_{(j, j)}=1$, for all $j=1,\dots,p$. We simulate observations $\bx_1,\ldots,\bx_n \sim_{i.i.d.}\mathcal{N}_p(\bzero, \bSigma)$, and generate the outcome $\by \sim \mathcal{N}_n(\bX\bbeta^\ast, \sigma^2_\bepsilon{\bf I}_n)$, $n \in \{100, 300, 500\}$, where 
$$\beta^\ast_j= \begin{cases}
1 & \mbox { for } j=1\\
 0.1 & \mbox { for } 2 \leq j \leq 5 \\
  0 & \mbox { otherwise } 
\end{cases}.
$$ 
A range of error variances $\sigma_\bepsilon^2$ are used to produce signal-to-noise ratios, $\textrm{SNR}\equiv(\bbeta^{\ast\top}\bSigma\bbeta^\ast)/\sigma^2_\bepsilon\in\{0.1, 0.3, 0.5\}$.

\begin{figure}[t]
\caption{The scale-free graph and stochastic block model settings. The size of a given node indicates the magnitude of the corresponding element of $\bbeta^\ast$.}
\centering\includegraphics[width = 12cm]{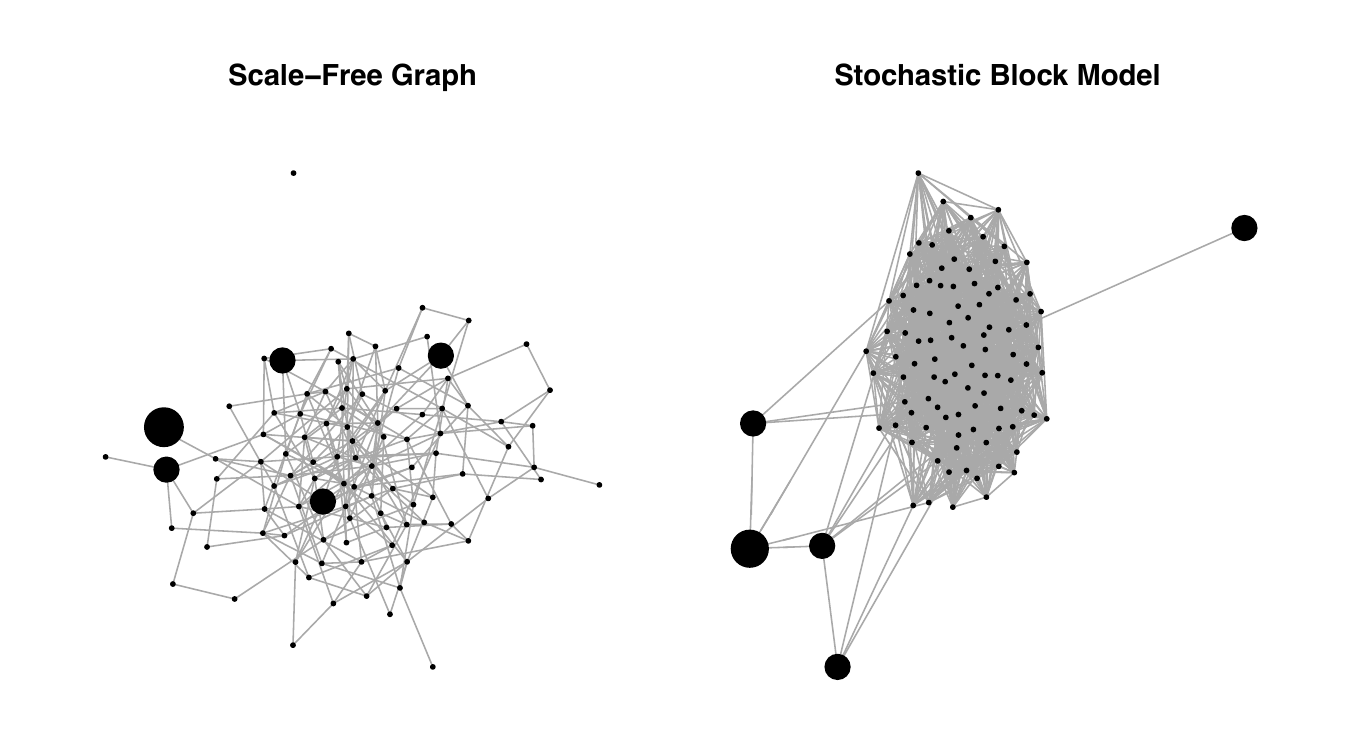}
\label{fig:graphsettings}
\end{figure}

Throughout the simulations, $\bSigma$ and $\bbeta^\ast$ are held fixed over $B=1000$ repetitions of the simulation study, while $\bX$ and $\by$ vary.

\subsection{Simulation Results}

We calculate the average length and coverage proportion of the 95\% na\"{i}ve confidence intervals, where the coverage proportion is defined as 
\begin{equation}
\text{Coverage Proportion} \equiv\sum_{b=1}^B\sum_{j\in\hat\A_\lambda^b}
1\left\{\beta_j^{(\hat\A_\lambda^b)}\in \CI^{(\hat\A_\lambda^b),b}_j\right\} \bigg/ \left|\hat\A_\lambda^b\right|,
\label{eq:cp}
\end{equation}
where $\hat\A_\lambda^b$ and $\CI_j^{(\hat\A_\lambda^b), b}$ are the set of variables selected by the lasso in the $b$th repetition, and the 95\% na\"ive confidence interval  \eqref{eq:naiveci} for the $j$th variable in the $b$th repetition, respectively. Recall that $\beta_j^{(\hat\A_\lambda^b)}$ was defined in \eqref{eq:betaM}.
 In order to calculate the average length and coverage proportion associated with the exact lasso post selection procedure of \citet{lee2015exact}, we replace $\CI^{(\hat\A_\lambda^b),b}_j$ in \eqref{eq:cp} with the confidence interval output by the \texttt{selectiveInference} \textsc{R} package.

Tables~\ref{tab:waldER} and \ref{tab:waldDSG35} show the coverage proportion and average length of 95\% na\"{i}ve confidence intervals and 95\% exact lasso post-selection confidence intervals under the scale-free graph and stochastic block model settings, respectively. The result shows that the coverage probability of the exact post-selection confidence interval is more correct than that of the na\"ive confidence interval when the data are small and signal is weak. But when the data are large and/or signal is relatively strong, both confidence intervals have approximately correct coverage. This corroborates  the findings in \citet{LeebetalPoSI2015}, in which the authors consider settings with $n=30$ and $p=10$. The coverage probability of the na\"ive confidence intervals with tuning parameter $\lambda_{\text{1SE}}$ is a bit smaller than the desired level, especially when the signal is weak. This may be due to the randomness in $\lambda_{\text{1SE}}$. The na\"ive approach of error variance estimation as in Theorem~\ref{THM:ERRORVAR} works similarly compared to the scaled lasso across all settings. In addition, Tables~\ref{tab:waldER} and \ref{tab:waldDSG35} also show that na\"ive confidence intervals are substantially narrower than exact lasso post-selection confidence intervals, especially when the signal is weak.

\begin{table}[h]
\footnotesize
\centering 
\caption{Coverage proportions (Cov) and average lengths (Len) of 95\% na\"{i}ve  confidence intervals with tuning parameters $\lambda_{\text{sup}}$ and $\lambda_{\text{1SE}}$,  and 95\% exact post-selection confidence intervals under the scale-free graph setting with partial correlation $\rho \in \{0.2, 0.6 \}$, sample size $n\in\{100, 300, 500\}$, dimension $p=100$ and signal-to-noise ratio $\mathrm{SNR}\in\{0.1, 0.3, 0.5\}$. The error variance is estimated either through the scaled lasso \citep{SunZhang2012ScaledLasso} (SL) or Theorem~\ref{THM:ERRORVAR} (NL).}
\begin{tabular}{cc|ccc|ccc|ccc} 
\hline
\hline
&$\rho$ & \multicolumn{9}{c}{0.2} \\
&$n$ & \multicolumn{3}{c|}{100} & \multicolumn{3}{c|}{300} & \multicolumn{3}{c}{500}\\
&\textrm{SNR} & $0.1$ & $0.3$ & $0.5$  & $0.1$ & $0.3$ & $0.5$  & $0.1$ & $0.3$ & $0.5$
\\
\hline 
\multirow{3}{*}{Cov}	&exact $\lambda_{\text{sup}}$ SL		& 0.905 & 0.942 & 0.936 & 0.949 & 0.953 & 0.953 & 0.946 & 0.949 & 0.949 \\
					&na\"{i}ve $\lambda_{\text{sup}}$ SL 	& 0.683 & 0.961 & 0.944 & 0.951 & 0.952 & 0.951 & 0.971 & 0.948 & 0.947\\
					&na\"{i}ve $\lambda_{\text{1SE}}$ SL	& 0.596 & 0.876 & 0.887 & 0.905 & 0.917 & 0.910 & 0.944 & 0.935 & 0.932\\
					&na\"{i}ve $\lambda_{\text{1SE}}$ NL 	& 0.587 & 0.870 & 0.883 & 0.905 & 0.919 & 0.911 & 0.945 & 0.939 & 0.933\\
\hline 
\multirow{3}{*}{Len}	&exact $\lambda_{\text{sup}}$	SL	& 4.260 & 1.634 & 0.823 & 1.907 & 0.437 & 0.330 & 0.814 & 0.327 & 0.255 \\
					&na\"{i}ve $\lambda_{\text{sup}}$ SL	& 1.122 & 0.693 & 0.552 & 0.717 & 0.421 & 0.326 & 0.562 & 0.325 & 0.253 \\
					&na\"{i}ve $\lambda_{\text{1SE}}$ SL	& 1.199 & 0.705 & 0.555 & 0.718 & 0.420 & 0.326 & 0.560 & 0.326 & 0.253\\
					&na\"{i}ve $\lambda_{\text{1SE}}$ NL 	& 1.201 & 0.711 & 0.559 & 0.723 & 0.424 & 0.329 & 0.563 & 0.328 & 0.255\\
\hline 
&$\rho$ & \multicolumn{9}{c}{0.6} \\
&$n$ & \multicolumn{3}{c|}{100}& \multicolumn{3}{c|}{300} & \multicolumn{3}{c}{500}\\
&\textrm{SNR} & $0.1$ & $0.3$ & $0.5$  & $0.1$ & $0.3$ & $0.5$  & $0.1$ & $0.3$ & $0.5$
\\
\hline 
\multirow{3}{*}{Cov}	&exact $\lambda_{\text{sup}}$ SL		& 0.956 & 0.938 & 0.945 & 0.922 & 0.944 & 0.945 & 0.948 & 0.945 & 0.947 \\
					&na\"{i}ve $\lambda_{\text{sup}}$ SL	& 0.625 & 0.943 & 0.953 & 0.942 & 0.942 & 0.942 & 0.964 & 0.944 & 0.944\\
					&na\"{i}ve $\lambda_{\text{1SE}}$ SL	& 0.643 & 0.873 & 0.876 & 0.918 & 0.934 & 0.929 & 0.954 & 0.941 & 0.935\\
					&na\"{i}ve $\lambda_{\text{1SE}}$ NL 	& 0.639 & 0.870 & 0.873 & 0.919 & 0.934 & 0.929 & 0.957 & 0.944 & 0.938\\
\hline 
\multirow{3}{*}{Len}	&exact $\lambda_{\text{sup}}$	SL	& 6.433 & 1.616 & 0.770 & 1.892 & 0.429 & 0.328 & 0.787 & 0.326 & 0.253 \\
					&na\"{i}ve $\lambda_{\text{sup}}$ SL	& 1.108 & 0.692 & 0.553 & 0.713 & 0.420 & 0.326 & 0.560 & 0.325 & 0.252\\
					&na\"{i}ve $\lambda_{\text{1SE}}$ SL 	& 1.200 & 0.705 & 0.554 & 0.720 & 0.420 & 0.326 & 0.561 & 0.326 & 0.253\\
					&na\"{i}ve $\lambda_{\text{1SE}}$ NL 	& 1.203 & 0.709 & 0.557 & 0.725 & 0.423 & 0.329 & 0.564 & 0.327 & 0.254\\
\hline \hline
\end{tabular}
\label{tab:waldER}
\end{table}

\begin{table}[h]
\footnotesize
\caption{Coverage proportions (Cov) and average lengths (Len) of 95\% na\"{i}ve  confidence intervals with tuning parameters $\lambda_{\text{sup}}$ and $\lambda_{\text{1SE}}$, and 95\% exact post-selection confidence intervals  under the stochastic block model setting with partial correlation $\rho \in \{0.2, 0.6 \}$, sample size $n\in\{100, 300, 500\}$, dimension $p=100$ and signal-to-noise ratio $\mathrm{SNR}\in\{0.1, 0.3, 0.5\}$. The error variance is estimated either through the scaled lasso \citep{SunZhang2012ScaledLasso} (SL) or Theorem~\ref{THM:ERRORVAR} (NL).}
\centering 
\begin{tabular}{cc|ccc|ccc|ccc} 
\hline
\hline
&$\rho$ & \multicolumn{9}{c}{0.2} \\
&$n$ & \multicolumn{3}{c|}{100} & \multicolumn{3}{c|}{300} & \multicolumn{3}{c}{500}\\
&\textrm{SNR} & $0.1$ & $0.3$ & $0.5$  & $0.1$ & $0.3$ & $0.5$  & $0.1$ & $0.3$ & $0.5$
\\
\hline 
\multirow{3}{*}{Cov}		&exact $\lambda_{\text{sup}}$ SL		& 0.921 & 0.934 & 0.948 & 0.951 & 0.934 & 0.933 & 0.953 & 0.958 & 0.958 \\
					&na\"{i}ve $\lambda_{\text{sup}}$ SL	& 0.540 & 0.955 & 0.960 & 0.953 & 0.931 & 0.930 & 0.966 & 0.957 & 0.957 \\
					&na\"{i}ve $\lambda_{\text{1SE}}$ SL	& 0.585 & 0.850 & 0.858 & 0.911 & 0.922 & 0.917 & 0.961 & 0.944 & 0.940 \\
					&na\"{i}ve $\lambda_{\text{1SE}}$ NL	& 0.586 & 0.854 & 0.865 & 0.913 & 0.924 & 0.918 & 0.961 & 0.945 & 0.942 \\
\hline
\multirow{3}{*}{Len}		&exact $\lambda_{\text{sup}}$ SL		& 4.721 & 1.692 & 0.857 & 1.873 & 0.425 & 0.325 & 0.758 & 0.325 & 0.252 \\
					&na\"{i}ve $\lambda_{\text{sup}}$ SL	& 1.111 & 0.682 & 0.545 & 0.709 & 0.416 & 0.323 & 0.559 & 0.323 & 0.251 \\
					&na\"{i}ve $\lambda_{\text{1SE}}$ SL 	& 1.188 & 0.703 & 0.552 & 0.715 & 0.416 & 0.323 & 0.557 & 0.323 & 0.251 \\
					&na\"{i}ve $\lambda_{\text{1SE}}$ NL 	& 1.195 & 0.710 & 0.558 & 0.721 & 0.420 & 0.326 & 0.560 & 0.325 & 0.253 \\
\hline 
&$\rho$ & \multicolumn{9}{c}{0.6} \\
&$n$ & \multicolumn{3}{c|}{100} & \multicolumn{3}{c|}{300} & \multicolumn{3}{c}{500}\\
&\textrm{SNR} & $0.1$ & $0.3$ & $0.5$  & $0.1$ & $0.3$ & $0.5$  & $0.1$ & $0.3$ & $0.5$
\\
\hline 
\multirow{3}{*}{Cov}		&exact $\lambda_{\text{sup}}$ SL		& 0.923 & 0.945 & 0.949 & 0.959 & 0.956 & 0.958 & 0.943 & 0.948 & 0.948 \\
					&na\"{i}ve $\lambda_{\text{sup}}$ SL	& 0.569 & 0.959 & 0.964 & 0.964 & 0.955 & 0.956 & 0.963 & 0.946 & 0.946 \\
					&na\"{i}ve $\lambda_{\text{1SE}}$ SL	& 0.568 & 0.841 & 0.852 & 0.922 & 0.933 & 0.927 & 0.947 & 0.946 & 0.941 \\
					&na\"{i}ve $\lambda_{\text{1SE}}$ NL	& 0.561 & 0.841 & 0.850 & 0.922 & 0.935 & 0.929 & 0.950 & 0.947 & 0.943 \\
\hline 
\multirow{3}{*}{Len}		&exact $\lambda_{\text{sup}}$ SL		& 4.744 & 1.662 & 0.832 & 1.987 & 0.424 & 0.325 & 0.807 & 0.324 & 0.251 \\
					&na\"{i}ve $\lambda_{\text{sup}}$ SL	& 1.119 & 0.690 & 0.545 & 0.712 & 0.416 & 0.323 & 0.557 & 0.322 & 0.250 \\
					&na\"{i}ve $\lambda_{\text{1SE}}$ SL	& 1.194 & 0.704 & 0.554 & 0.713 & 0.416 & 0.323 & 0.555 & 0.322 & 0.250 \\
					&na\"{i}ve $\lambda_{\text{1SE}}$ NL	& 1.194 & 0.707 & 0.557 & 0.718 & 0.420 & 0.326 & 0.558 & 0.324 & 0.252 \\
\hline \hline
\end{tabular}
\label{tab:waldDSG35}
\end{table}

In addition, to evaluate whether $\hat\A_\lambda$ is deterministic, among repetitions that $\hat\A_\lambda^b\neq\emptyset$ (there is no confidence interval if $\hat\A_\lambda^b=\emptyset$), we also calculate the proportion of $\hat\A_\lambda^b=\cD$, where $\cD$ is the most common $\hat\A_\lambda^b$, $b=1,\dots,1000$. The result is summarized in Table~\ref{tab:waldfixed}, which shows that $\hat\A_\lambda$ is almost deterministic with tuning parameter $\lambda_{\text{sup}}$. With $\lambda_{\text{1SE}}$, due to the randomness in the tuning parameter, $\hat\A_\lambda$ is less deterministic, which may explain the smaller coverage probability than the desired level in this case. We also examined the proportion of $\hat\A_\lambda^b=\A_\lambda^b$, which is shown in Table~\ref{tab:equalalambda}. The observation is similar to that in Table~\ref{tab:waldfixed}. As mentioned in Section~\ref{sec:intro}, $\Pr[\hat\A_\lambda = \A^\ast] = 0$ in all the settings considered. 

\begin{table}[h]
\footnotesize
\caption{Among repetitions that $\hat\A_\lambda^b\neq\emptyset$, the proportion of $\hat\A_\lambda^b$ that equals the most common $\hat\A_\lambda^b$, $b=1,\dots,1000$, under the scale-free graph and stochastic block model settings with tuning parameters $\lambda_{\text{sup}}$ and $\lambda_{\text{1SE}}$. In the simulation, $\rho \in \{0.2, 0.6 \}$, sample size $n\in\{100, 300, 500\}$, dimension $p=100$ and signal-to-noise ratio $\mathrm{SNR}\in\{0.1, 0.3, 0.5\}$.}
\centering 
\begin{tabular}{c|ccc|ccc|ccc} 
\hline
\hline
$\rho$ & \multicolumn{9}{c}{0.2} \\
$n$ & \multicolumn{3}{c|}{100} & \multicolumn{3}{c|}{300} & \multicolumn{3}{c}{500}\\
\textrm{SNR} & $0.1$ & $0.3$ & $0.5$  & $0.1$ & $0.3$ & $0.5$  & $0.1$ & $0.3$ & $0.5$
\\
\hline 
Scale-free; $\lambda_{\text{sup}}$ 	& 0.000 & 0.999 & 0.993 & 1.000 & 0.999 & 0.998 &1.000 & 0.999 & 0.996 \\
Scale-free; $\lambda_{\text{1SE}}$ 	& 0.832 & 0.958 & 0.851 & 0.961 & 0.948 & 0.934 & 0.988 & 0.980 & 0.966 \\
\hline 
Stochastic block; $\lambda_{\text{sup}}$	&1.000 & 0.997 & 0.992 & 0.998 & 1.000 & 0.999 & 1.000 & 1.000 & 1.000 \\
Stochastic block; $\lambda_{\text{1SE}}$	& 0.838 & 0.839 & 0.831 & 0.964 & 0.957 & 0.946 & 0.992 & 0.979 & 0.968 \\
\hline 
$\rho$ & \multicolumn{9}{c}{0.6} \\
$n$ & \multicolumn{3}{c|}{100} & \multicolumn{3}{c|}{300} & \multicolumn{3}{c}{500}\\
\textrm{SNR} & $0.1$ & $0.3$ & $0.5$  & $0.1$ & $0.3$ & $0.5$  & $0.1$ & $0.3$ & $0.5$
\\
\hline 
Scale-free; $\lambda_{\text{sup}}$	& 0.995 & 0.996 & 0.996 & 0.999 & 0.998 & 0.998 & 0.999 & 1.000 & 1.000 \\
Scale-free; $\lambda_{\text{1SE}}$	& 0.854 & 0.869 & 0.855 & 0.964 & 0.947 & 0.935 & 0.994 & 0.981 & 0.962 \\
\hline 
Stochastic block; $\lambda_{\text{sup}}$	& 0.998 & 0.998 & 0.998 & 1.000 & 1.000 & 1.000 & 1.000 & 0.999 & 0.999 \\
Stochastic block; $\lambda_{\text{1SE}}$	& 0.834 & 0.843 & 0.826 & 0.969 & 0.976 & 0.968 & 0.988 & 0.979 & 0.971 \\
\hline \hline
\end{tabular}
\label{tab:waldfixed}
\end{table}

\begin{table}[h]
\footnotesize
\caption{The proportion of $\hat\A_\lambda^b$ that equals $\A_\lambda^b$, $b=1,\dots,1000$, under the scale-free graph and stochastic block model settings with tuning parameters $\lambda_{\text{sup}}$ and $\lambda_{\text{1SE}}$. In the simulation, $\rho \in \{0.2, 0.6 \}$, sample size $n\in\{100, 300, 500\}$, dimension $p=100$ and signal-to-noise ratio $\mathrm{SNR}\in\{0.1, 0.3, 0.5\}$.}
\centering 
\begin{tabular}{c|ccc|ccc|ccc} 
\hline
\hline
$\rho$ & \multicolumn{9}{c}{0.2} \\
$n$ & \multicolumn{3}{c|}{100} & \multicolumn{3}{c|}{300} & \multicolumn{3}{c}{500}\\
\textrm{SNR} & $0.1$ & $0.3$ & $0.5$  & $0.1$ & $0.3$ & $0.5$  & $0.1$ & $0.3$ & $0.5$
\\
\hline 
Scale-free $\lambda_{\text{sup}}$ 	& 0.958 & 0.725 & 0.928 & 0.649 & 1.000 & 1.000 &0.971 & 0.999 & 0.998 \\
Scale-free $\lambda_{\text{1SE}}$ 	& 0.417 & 0.665 & 0.829 & 0.567 & 0.955 & 0.940 & 0.804 & 0.976 & 0.964 \\
\hline 
Stochastic block $\lambda_{\text{sup}}$	&0.964 & 0.736 & 0.931 & 0.678 & 1.000 & 1.000 & 0.964 & 1.000 & 1.000 \\
Stochastic block $\lambda_{\text{1SE}}$	& 0.394 & 0.661 & 0.815 & 0.582 & 0.964 & 0.957 & 0.582 & 0.559 & 0.964 \\
\hline 
$\rho$ & \multicolumn{9}{c}{0.6} \\
$n$ & \multicolumn{3}{c|}{100} & \multicolumn{3}{c|}{300} & \multicolumn{3}{c}{500}\\
\textrm{SNR} & $0.1$ & $0.3$ & $0.5$  & $0.1$ & $0.3$ & $0.5$  & $0.1$ & $0.3$ & $0.5$
\\
\hline 
Scale-free $\lambda_{\text{sup}}$	& 0.952 & 0.713 & 0.924 & 0.687 & 1.000 & 1.000 & 0.968 & 0.999 & 0.998 \\
Scale-free $\lambda_{\text{1SE}}$	& 0.396 & 0.657 & 0.799 & 0.572 & 0.950 & 0.941 & 0.779 & 0.976 & 0.963 \\
\hline 
Stochastic block $\lambda_{\text{sup}}$	& 0.950 & 0.750 & 0.935 & 0.687 & 1.000 & 1.000 & 0.967 & 1.000 & 0.999 \\
Stochastic block $\lambda_{\text{1SE}}$	& 0.445 & 0.661 & 0.806 & 0.559 & 0.959 & 0.948 & 0.805 & 0.987 & 0.974 \\
\hline \hline
\end{tabular}
\label{tab:equalalambda}
\end{table}

\section{Inference for $\beta^\ast$ With the Na\"{i}ve Score Test}\label{sec:waldscore}

Sections~\ref{sec:nllasso} and \ref{sec:waldsims} focused on the task of developing confidence intervals for  $\bbeta^{(\cM)}$ in \eqref{eq:submodel}, where $\cM=\hat\A_\lambda$, the set of variables selected by the lasso.  However, recall from \eqref{eq:betaM} that typically 
$\bbeta^{(\cM)} \neq \bbeta^\ast_{\cM}$, where $\bbeta^\ast$ was introduced in \eqref{eq:linmodel}. 

In this section, we shift our focus to performing inference on $\bbeta^\ast$. 
 We will exploit Proposition~\ref{THM:CONSISTENT} to develop a simple approach for testing $H_{0,j}^\ast:\beta_j^\ast=0$,   for  $j=1,\ldots,p$.

 Recall that in the low-dimensional setting, the classical score statistic for the hypothesis $H_{0,j}^\ast: \beta_j^\ast=0$ is proportional to $\bx_j^T(\by-\hat\by^0)$, where $\hat\by^0$ is the vector of  fitted values that results from least squares regression of $\by$ onto the $p-1$ features $\bx_1,\ldots,\bx_{j-1},\bx_{j+1},\ldots,\bx_p$. 
 In order to adapt the classical score test statistic to the high-dimensional setting, we define the \emph{na\"{i}ve score test statistic} for testing $H_{0,j}^\ast: \beta_j^\ast=0$ as 
 \begin{equation}
 S^{j} 
 \equiv \bx_j^\top\left(\by-\tilde\by^{(\hat \A_\lambda\backslash\{j\})}\right) \equiv \bx_j^\top\left({\bf I}_n -{\bf P}^{-j}\right) \by,
 \label{eq:lststat}
 \end{equation}
 where  
\begin{equation*}\tilde\by^{(\hat \A_\lambda\backslash\{j\})}\equiv\bX_{\hat\A_\lambda\backslash\{j\}}\tilde\bbeta^{(\hat\A_\lambda\backslash\{j\})},
\end{equation*}
and
\begin{equation*}
{\bf P}^{-j}\equiv\bX_{\hat \A_\lambda\backslash\{j\}}\left(\bX_{\hat \A_\lambda\backslash\{j\}}^\top\bX_{\hat \A_\lambda\backslash\{j\}}\right)^{-1}\bX_{\hat \A_\lambda\backslash\{j\}}^\top
\label{eq:posiolsscore}
\end{equation*}
is the orthogonal projection matrix onto the set of variables in ${\hat \A_\lambda\backslash\{j\}}$. $\tilde\bbeta^{(\hat\A_\lambda\backslash\{j\})}$ is defined in \eqref{eq:posiols}.
In \eqref{eq:lststat}, 
  the notation $\hat\A_\lambda\backslash \{j\}$ represents the set $\hat\A_\lambda$ in \eqref{Alambda} with $j$ removed, if $j\in\hat\A_\lambda$. If $j \notin \hat\A_\lambda$, then  $\hat\A_\lambda\backslash \{j\} = \hat\A_\lambda$.

 In Theorem~\ref{THM:SCORE}, we will derive the asymptotic distribution of $S^{j}$ under $H_{0,j}^\ast: \beta_j^\ast=0$. We first introduce two new conditions. 
 
 First, we require that the total  signal strength of variables not selected by the noiseless lasso, \eqref{eq:nllasso}, is small.
 \begin{itemize}
\item[({\bf M2$^\ast$})] Recall that $\A^\ast\equiv\supp(\bbeta^\ast)$.  Let $\A_\lambda\equiv\supp(\bbeta_\lambda)$, $b_{\lambda\min}\equiv\min_{j\in\A_\lambda}|\beta_{\lambda,j}|$ with $\bbeta_\lambda$ defined in \eqref{eq:nllasso}. The signal strength in $\A_\lambda$ satisfies
\begin{align}
\lim_{n\to\infty}\frac{b_{\lambda\min}}{\lambda} =\xi > 0,
\end{align}
and the signal strength outside of $\A_\lambda$ satisfies
\begin{align}
\left\|\bX_{\A_\lambda^c}\bbeta^\ast_{\A_\lambda^c}\right\|_2=\smallO\left(1\right).
\end{align}
\end{itemize}
Condition ({\bf M2$^\ast$}) closely resembles ({\bf M2}), which was required for Theorem~\ref{THM:WALD2} in Section~\ref{sec:nllasso}. 
  The only difference between the two is that  ({\bf M2$^\ast$}) requires $\|\bX_{\A_\lambda^c}\bbeta^\ast_{\A_\lambda^c}\|_2=\smallO(1)$, whereas ({\bf M2}) requires only that $\|\bX_{\A_\lambda^c}\bbeta^\ast_{\A_\lambda^c}\|_2=\mathcal{O}\left(\sqrt{\log(p)}\right)$. Recall that in Section~\ref{sec:nllasso}, we consider inference for the parameters in the sub-model \eqref{eq:submodel}. In other words, testing the population regression parameter $\bbeta^\ast$ in \eqref{eq:linmodel} requires more stringent assumptions than constructing confidence intervals for the parameters in the sub-model \eqref{eq:submodel}. 

The following condition, required to apply the Lindeberg-Feller Central Limit Theorem, can  be relaxed if the noise $\bepsilon$ in \eqref{eq:linmodel} is normally distributed.
\begin{itemize}
\item[({\bf S})] $\lambda$, $\bbeta^\ast$ and $\bX$ satisfy $\lim_{n\to\infty}\| {\bf r}^s \|_\infty/ \|{\bf r}^s\|_2 = 0$, where ${\bf r}^s \equiv  \left({\bf I}_n - {\bf P}^{(\A_\lambda\backslash\{j\})} \right)\bx_j$. 
\end{itemize}

We now present Theorem~\ref{THM:SCORE}, which is proven in Section~\ref{sec:pfscore} of the online Supplementary Materials. 
\begin{theorem}\label{THM:SCORE}
Suppose ({\bf M2$^\ast$}) and ({\bf S}) hold and $\lim_{n\to\infty}\Pr\left[\hat \A_\lambda= \A_\lambda\right]= 1$. For any $j=1,\dots,p$, under the null hypothesis $H_{0,j}^\ast:\beta^\ast_j=0$,
\begin{align}\label{eq:Ts}
T\equiv\frac{S^{j}}{\sigma_\bepsilon\sqrt{\bx_j^\top\left( {\bf I}_n - {\bf P}^{-j}\right)\bx_j}}\rightarrow_{d} \mathcal{N}(0, 1),
\end{align}
where  $S^{j}$ was defined in \eqref{eq:lststat}, and where $\sigma_\bepsilon$ is the variance of $\bepsilon$ in \eqref{eq:linmodel}. 
\end{theorem}
Theorem~\ref{THM:SCORE} states that the distribution of the na\"{i}ve score test statistic $S^{j}$ is asymptotically the same as if $\hat\A_\lambda$ were a fixed set, as opposed to being selected by fitting a lasso model on  the data.   Based on \eqref{eq:Ts}, we reject the null hypothesis $H^\ast_{0,j}:\beta^\ast_j=0$ at level $\alpha>0$ if $|T|>\Phi_{\mathcal{N}}^{-1}(1-\alpha/2)$, where $\Phi^{-1}_{\mathcal{N}}(\cdot)$ is the quantile function of the standard normal distribution function. 

We emphasize that Theorem~\ref{THM:SCORE} holds for any variable $j=1,\dots,p$, and thus can be used to test $H_{0,j}^\ast: \beta_j^\ast = 0$,  for all $j=1,\dots,p$. (This is in contrast to  Theorem~\ref{THM:WALD2}, which concerns  confidence intervals for the parameters in the sub-model \eqref{eq:submodel} consisting of the variables in $\hat\A_\lambda$, and hence holds only for  $j\in\hat\A_\lambda$.)

\section{Numerical Examination of the Na\"{i}ve Score Test}\label{sec:numeric}

In this section, we compare the performance of the na\"{i}ve score test \eqref{eq:lststat}  to three recent proposals from the literature for testing 
$H_{0,j}^\ast: \beta_j^\ast=0$: namely, 
LDPE \citep{ZhangZhang2014LDPE, vandeGeeretal2014LDPE}, SSLasso \citep{javanmard2013confidence}, and the decorrelated score test \citep[dScore;][]{NingLiu2015decor}. (Since with high probability $\hat\A_\lambda \ne \A^\ast$, we do not include the exact post-selection procedure in this comparison.)
 \textsc{R} code for SSLasso, and dScore was provided by the authors; LDPE is implemented in the \textsc{R} package \texttt{hdi}.
 For the na\"{i}ve score test, we estimate $\sigma_\bepsilon$, the standard deviation of the errors in \eqref{eq:linmodel},   using the scaled lasso \citep{SunZhang2012ScaledLasso} or Theorem~\ref{THM:ERRORVAR}.

 All four of these methods require us to select the value of the lasso  tuning parameter. For LDPE, SSLasso, and dScore, we use 10-fold cross-validation to select the tuning parameter value that produces the smallest cross-validated mean square error, $\lambda_{\text{min}}$. As in the numerical study of the na\"{i}ve confidence intervals in Section~\ref{sec:waldsims}, we implement the na\"{i}ve score test using the tuning parameter value $\lambda_{\text{1SE}}$ and $\lambda_{\text{sup}}$. Unless otherwise noted, all tests are performed at a significance level of 0.05.

 In  Section~\ref{sec:power}, we investigate the powers and type-I errors of the above tests in simulation experiments. 
   Section~\ref{sec:realdata} contains an analysis of a glioblastoma gene expression dataset.


\subsection{Power and Type-I Error} \label{sec:power} 

\subsubsection{Simulation Set-Up}

In this section, we adapt the scale-free graph and the stochastic block model presented in  Section~\ref{sec:simsetup-Sec3}  to have $p=500$.

In the scale-free graph setting, we generate a scale-free graph with  $\gamma=5$, edge density 0.05, and $p=500$ nodes. The resulting graph has $|\cE|=6237$ edges. We order the nodes in the graph so that $j$th node is the $(30 \times j)$th least-connected node in the graph, for $1 \leq j \leq 10$. For example, the 4th node is the 120th least-connected node in the graph. 

In the stochastic block model setting, we generate two dense Erd\H{o}s-R\'enyi graphs with ten nodes and 490 nodes, respectively; each has an intra-graph edge density of 0.3. The node indices are ordered so that the nodes in the smaller graph precede those in the larger graph. 
 We then randomly connect nodes between the  two graphs in order to obtain an inter-graph edge density of 0.05. 

Next, for both graph settings, we generate $\bA$ as in \eqref{eq:A},
 where 
 $\rho\in\{0.2,0.6\}$. 
We then set $\bSigma=\bA^{-1}$, and standardize $\bSigma$ so that $\Sigma_{(j, j)}=1$, for all $j=1,\dots,p$. We simulate observations $\bx_1,\ldots,\bx_n \sim_{i.i.d.}\mathcal{N}_p(\bzero, \bSigma)$, and generate the outcome $\by \sim \mathcal{N}_n(\bX\bbeta^\ast, \sigma^2_\bepsilon{\bf I}_n)$, $n \in \{100, 200, 400\}$, where 
$$\beta^\ast_j= \begin{cases}
1 & \mbox { for } 1 \leq j \leq 3\\
 0.1 & \mbox { for } 4 \leq j \leq 10 \\
  0 & \mbox { otherwise } 
\end{cases}.
$$ 
A range of error variances $\sigma_\bepsilon^2$ are used to produce signal-to-noise ratios, $\textrm{SNR}\equiv(\bbeta^{\ast\top}\bSigma\bbeta^\ast)/\sigma^2_\bepsilon\in\{0.1, 0.3, 0.5\}$.

We hold $\bSigma$ and $\bbeta^\ast$ fixed over $B=100$ repetitions of the simulation, while $\bX$ and $\by$ vary.



\subsubsection{Simulation Results}

For each test, the average power  on the strong signal variables, the average power on the weak signal variables, and  the average type-I error rate are defined as 
\begin{align}
\text{Power}_\text{strong}  &\equiv \frac{1}{B}\frac{1}{3} \sum_{b=1}^B \sum_{j:\beta_j^\ast=1} 1\{p_{jb} < 0.05 \}, \label{powers} \\
\text{Power}_\text{weak}  &\equiv \frac{1}{B}\frac{1}{7} \sum_{b=1}^B \sum_{j:\beta_j^\ast=0.1} 1\{p_{jb} < 0.05 \}, \label{powew} \\
\text{Type-1 Error} &\equiv \frac{1}{B}\frac{1}{490}\sum_{b=1}^B \sum_{j:\beta_j^\ast=0} 1\{p_{jb} < 0.05 \}, \label{T1err} 
\end{align}
respectively. 
In \eqref{powers}--\eqref{T1err}, $p_{jb}$ is the $p$-value associated with null hypothesis $H^\ast_{0, j}:\beta_j^\ast=0$ in the $b$th simulated data set. In the simulations, the graphs and $\bbeta^\ast$ are held fixed over $B=100$ repetitions of the simulation study, while $\bX$ and $\by$ vary.  

Tables~\ref{tab:scoreER} and \ref{tab:scoreDSG35} summarize the results in the two simulation settings. Na\"ive score test with $\lambda_{\text{sup}}$ has slightly worse control of type-1 error rate and better power than the other four methods, which have approximate control over the type-I error rate and comparable power. The performance with the scaled lasso is similar to the performance of Theorem~\ref{THM:ERRORVAR}.

\begin{table}[h]
\caption{Average power and type-I error rate for the  hypotheses $H_{0,j}^\ast:\beta_j^\ast=0$ for $j=1,\dots,p$, as defined in \eqref{powers}--\eqref{T1err},    under the scale-free graph  setting with $p=500$.  Results are shown for various values of $\rho$, $n$, $\textrm{SNR}$.  Methods for comparison include LDPE, SSLasso, dScore, and the na\"{i}ve score test with tuning parameter $\lambda_{\text{min}}$, $\lambda_{\text{1SE}}$ and $\lambda_{\text{sup}}$. The error variance is estimated either through the scaled lasso \citep{SunZhang2012ScaledLasso} (S-Z) or Theorem~\ref{THM:ERRORVAR} (T2.2). Note that as mentioned in Section~\ref{sec:waldsims}, we do not combine $\lambda_{\text{sup}}$ with T2.2}
\footnotesize
\centering 
\begin{tabular}{c|c|ccc|ccc|ccc} 
\hline
\hline
& $\rho$ & \multicolumn{9}{c}{0.2} \\
& $n$ & \multicolumn{3}{c|}{100} & \multicolumn{3}{c|}{200} & \multicolumn{3}{c}{400}\\
& \textrm{SNR} & $0.1$ & $0.3$ & $0.5$ & $0.1$ & $0.3$ & $0.5$  & $0.1$ & $0.3$ & $0.5$
\\ 
\hline 
\multirow{5}{*}{$\text{Power}_\text{strong}$}	& LDPE $\lambda_{\text{min}}$	 S-Z		&0.400 & 0.773 & 0.910 & 0.627 & 0.973 & 1.000 & 0.923 & 1.000 & 1.000 \\
  									& SSLasso $\lambda_{\text{min}}$ S-Z	&0.410 & 0.770 & 0.950 & 0.650 & 0.970  & 1.000 & 0.910 & 1.000 & 1.000 \\
 									& dScore $\lambda_{\text{min}}$ S-Z		&0.330 & 0.643 & 0.857 & 0.547 & 0.957 & 1.000 & 0.887 & 1.000 & 1.000 \\
 									& nScore $\lambda_{\text{sup}}$ S-Z 	&0.403 & 0.847 & 0.960 & 0.727 & 0.990 & 0.997 & 0.940 & 1.000 & 1.000 \\
 									& nScore $\lambda_{\text{1SE}}$ S-Z 	&0.427 & 0.763 & 0.893 & 0.677 & 0.977 & 1.000 & 0.957 & 1.000 & 1.000 \\
 									& nScore $\lambda_{\text{1SE}}$ T2.2 	&0.357 & 0.763 & 0.890 & 0.680 & 0.977 & 1.000 & 0.910 & 1.000 & 1.000 \\
\hline
\multirow{5}{*}{$\text{Power}_\text{weak}$}	& LDPE $\lambda_{\text{min}}$ S-Z		&0.064 & 0.083 & 0.056 & 0.054 & 0.059 & 0.079 & 0.070 & 0.079 & 0.113 \\
 									& SSLasso $\lambda_{\text{min}}$ S-Z	&0.081 & 0.087 & 0.060 & 0.066 & 0.061 & 0.086 & 0.069 & 0.086 & 0.113  \\
									& dScore $\lambda_{\text{min}}$ S-Z		&0.044 & 0.056 & 0.036 & 0.039 & 0.039 & 0.060 & 0.046 & 0.056 & 0.093 \\
 									& nScore $\lambda_{\text{sup}}$ S-Z 	&0.061 & 0.091 & 0.109 & 0.070 & 0.109 & 0.107 & 0.097 & 0.103 & 0.114 \\
 									& nScore $\lambda_{\text{1SE}}$ S-Z 	&0.080 & 0.077 & 0.059 & 0.060 & 0.061 &  0.061 & 0.083 & 0.076 & 0.101 \\
 									& nScore $\lambda_{\text{1SE}}$ T2.2 	&0.067 & 0.070 & 0.070 & 0.054 & 0.071 &  0.069 & 0.079 & 0.094 & 0.123 \\
\hline
\multirow{5}{*}{T1 Error} 					& LDPE $\lambda_{\text{min}}$	 S-Z		&0.051 & 0.052 & 0.051 & 0.049 & 0.051 & 0.047 & 0.050 & 0.051 & 0.049 \\
 									& SSLasso $\lambda_{\text{min}}$ S-Z	&0.056 & 0.056 & 0.056 & 0.054 & 0.055 & 0.053 & 0.053 & 0.054 & 0.054 \\
									& dScore $\lambda_{\text{min}}$ S-Z 	&0.035 & 0.040 & 0.040 & 0.033 & 0.036 & 0.034 & 0.035 & 0.037 & 0.034  \\
 									& nScore $\lambda_{\text{sup}}$ S-Z		&0.069 & 0.082 & 0.095 & 0.064 & 0.083 & 0.079 & 0.065 & 0.068 & 0.050  \\
 									& nScore $\lambda_{\text{1SE}}$ S-Z 	&0.061 & 0.057 & 0.048 & 0.056 & 0.055 & 0.040 & 0.060 & 0.046 & 0.046  \\
 									& nScore $\lambda_{\text{1SE}}$ T2.2 	&0.049 & 0.052 & 0.049 & 0.054 & 0.053 & 0.050 & 0.056 & 0.049 & 0.050  \\
\hline 
\hline
& $\rho$ & \multicolumn{9}{c}{0.6} \\
& $n$ & \multicolumn{3}{c|}{100} & \multicolumn{3}{c|}{200} & \multicolumn{3}{c}{400}\\
& \textrm{SNR} & $0.1$ & $0.3$ & $0.5$ & $0.1$ & $0.3$ & $0.5$  & $0.1$ & $0.3$ & $0.5$
\\ 
\hline 
\multirow{5}{*}{$\text{Power}_\text{strong}$} 	& LDPE $\lambda_{\text{min}}$	 S-Z		&0.330 & 0.783 & 0.947 & 0.627 & 0.980 & 1.000 & 0.887 & 1.000 & 1.000 \\
  									& SSLasso $\lambda_{\text{min}}$ S-Z	&0.347 & 0.790 & 0.957 & 0.623 & 0.987  & 1.000 & 0.867 & 1.000 & 1.000 \\
									& dScore $\lambda_{\text{min}}$ S-Z		&0.270 & 0.673 & 0.883 & 0.533 & 0.960 & 0.993 & 0.863 & 1.000 & 1.000 \\
 									& nScore $\lambda_{\text{sup}}$ S-Z		&0.430 & 0.790 & 0.933 & 0.707 & 0.977 & 1.000 & 0.923 & 1.000 & 1.000 \\
 									& nScore $\lambda_{\text{1SE}}$ S-Z	&0.357 & 0.767 & 0.887 & 0.677 & 0.980 & 0.997 & 0.937 & 1.000 & 1.000 \\
 									& nScore $\lambda_{\text{1SE}}$ T2.2	&0.340 & 0.697 & 0.907 & 0.637 & 0.973 & 1.000 & 0.950 & 1.000 & 1.000 \\
\hline
\multirow{5}{*}{$\text{Power}_\text{weak}$}	& LDPE $\lambda_{\text{min}}$ S-Z		&0.031 & 0.046 & 0.063 & 0.064 & 0.074 & 0.076 & 0.054 & 0.077 & 0.119 \\
 									& SSLasso $\lambda_{\text{min}}$ S-Z	&0.047 & 0.063 & 0.076 & 0.063 & 0.090 & 0.099 & 0.053 & 0.076 & 0.121  \\
									& dScore $\lambda_{\text{min}}$ S-Z		&0.021 & 0.037 & 0.047 & 0.036 & 0.060 & 0.044 & 0.034 & 0.050 & 0.083 \\
 									& nScore $\lambda_{\text{sup}}$ S-Z 	&0.071 & 0.089 & 0.136 & 0.081 & 0.121 &  0.104 & 0.114 & 0.113 & 0.123 \\
 									& nScore $\lambda_{\text{1SE}}$ S-Z 	&0.039 & 0.060 & 0.050 & 0.076 & 0.074 &  0.066 & 0.070 & 0.067 & 0.104 \\
 									& nScore $\lambda_{\text{1SE}}$T2.2 	&0.056 & 0.070 & 0.073 & 0.093 & 0.087 &  0.080 & 0.107 & 0.096 & 0.111 \\
\hline
\multirow{5}{*}{T1 Error}					& LDPE $\lambda_{\text{min}}$ S-Z	 	&0.050 & 0.051 & 0.051 & 0.050 & 0.049 & 0.051 & 0.051 & 0.050 & 0.047 \\
 									& SSLasso $\lambda_{\text{min}}$ S-Z	&0.056 & 0.056 & 0.056 & 0.054 & 0.055 & 0.053 & 0.053 & 0.054 & 0.054 \\
									& dScore $\lambda_{\text{min}}$ S-Z		&0.033 & 0.036 & 0.034 & 0.031 & 0.031 & 0.035 & 0.036 & 0.035 & 0.033  \\
 									& nScore $\lambda_{\text{sup}}$ S-Z 	&0.065 & 0.080 & 0.093 & 0.064 & 0.084 & 0.088 & 0.070 & 0.071 & 0.054  \\
 									& nScore $\lambda_{\text{1SE}}$ S-Z 	&0.056 & 0.060 & 0.045 & 0.061 & 0.051 & 0.040 & 0.058 & 0.048 & 0.047  \\
 									& nScore $\lambda_{\text{1SE}}$ T2.2 	&0.049 & 0.053 & 0.051 & 0.060 & 0.061 & 0.049 & 0.066 & 0.049 & 0.048  \\
\hline \hline
\end{tabular}
\label{tab:scoreER}
\end{table}

\begin{table}[h]
\caption{Average power and type-I error rate for the  hypotheses $H_{0,j}^\ast:\beta_j^\ast=0$ for $j=1,\dots,p$, as defined in \eqref{powers}--\eqref{T1err},    under the stochastic block model  setting with $p=500$. Details are as in Table~\ref{tab:scoreER}. }
\footnotesize
\centering 
\begin{tabular}{c|c|ccc|ccc|ccc} 
\hline
\hline
& $\rho$ & \multicolumn{9}{c}{0.2} \\
& $n$ & \multicolumn{3}{c|}{100} & \multicolumn{3}{c|}{200} & \multicolumn{3}{c}{400}\\
& \textrm{SNR} & $0.1$ & $0.3$ & $0.5$ & $0.1$ & $0.3$ & $0.5$  & $0.1$ & $0.3$ & $0.5$
\\ 
\hline 
\multirow{5}{*}{$\text{Power}_\text{strong}$} 		& LDPE $\lambda_{\text{min}}$	 S-Z		&0.370 & 0.793 & 0.937 & 0.687 & 0.990 & 1.000 & 0.914 & 1.000 & 1.000 \\
  										& SSLasso $\lambda_{\text{min}}$ S-Z	&0.393 & 0.803 & 0.933 & 0.687 & 0.990 & 1.000 & 0.892 & 1.000 & 1.000 \\
										& dScore $\lambda_{\text{min}}$ S-Z		&0.333 & 0.783 & 0.917 & 0.693 & 0.993 & 1.000 & 0.905 & 1.000 & 1.000 \\
 										& nScore $\lambda_{\text{sup}}$ S-Z 	&0.473 & 0.857 & 0.953 & 0.697 & 0.993 & 0.993 & 0.943 & 1.000 & 1.000 \\
 										& nScore $\lambda_{\text{1SE}}$ S-Z 	&0.400 & 0.797 & 0.903 & 0.713 & 0.997 & 1.000 & 0.910 & 1.000 & 1.000 \\
 										& nScore $\lambda_{\text{1SE}}$ T2.2 	&0.437 & 0.787 & 0.923 & 0.720 & 0.987 & 1.000 & 0.940 & 1.000 & 1.000 \\
\hline
\multirow{5}{*}{$\text{Power}_\text{weak}$}		& LDPE $\lambda_{\text{min}}$ S-Z		&0.041 & 0.044 & 0.051 & 0.057 & 0.050 & 0.071 & 0.050 & 0.093 & 0.071 \\
 										& SSLasso $\lambda_{\text{min}}$  S-Z	&0.054 & 0.056 & 0.074 & 0.071 & 0.056 & 0.089 & 0.071 & 0.101 & 0.101  \\
										& dScore $\lambda_{\text{min}}$ S-Z		&0.037 & 0.044 & 0.057 & 0.060 & 0.046 & 0.077 & 0.056 & 0.101 & 0.094 \\
 										& nScore $\lambda_{\text{sup}}$  S-Z	&0.059 & 0.071 & 0.107 & 0.043 & 0.083 & 0.070 & 0.069 & 0.094 & 0.106 \\
 										& nScore $\lambda_{\text{1SE}}$ S-Z	&0.047 & 0.059 & 0.060 & 0.059 & 0.047 & 0.059 & 0.062 & 0.106 & 0.105 \\
 										& nScore $\lambda_{\text{1SE}}$ T2.2	&0.057 & 0.064 & 0.067 & 0.047 & 0.069 & 0.086 & 0.043 & 0.079 & 0.113 \\
\hline
\multirow{5}{*}{T1ER} 						& LDPE $\lambda_{\text{min}}$	 S-Z		&0.051 & 0.049 & 0.048 & 0.050 & 0.050 & 0.050 & 0.051 & 0.050 & 0.049 \\
 										& SSLasso $\lambda_{\text{min}}$ S-Z	&0.057 & 0.056 & 0.058 & 0.054 & 0.054 & 0.054 & 0.054 & 0.053 & 0.054 \\
										& dScore $\lambda_{\text{min}}$ S-Z 	&0.043 & 0.040 & 0.041 & 0.041 & 0.044 & 0.042 & 0.042 & 0.042 & 0.041  \\
 										& nScore $\lambda_{\text{sup}}$  S-Z	&0.064 & 0.074 & 0.090 & 0.059 & 0.076 & 0.076 & 0.060 & 0.060 & 0.049 \\
 										& nScore $\lambda_{\text{1SE}}$ S-Z	&0.062 & 0.058 & 0.048 & 0.056 & 0.052 & 0.040 & 0.054 & 0.047 & 0.046  \\
 										& nScore $\lambda_{\text{1SE}}$ T2.2	&0.052 & 0.050 & 0.050 & 0.050 & 0.050 & 0.049 & 0.049 & 0.047 & 0.047  \\
\hline 
\hline
& $\rho$ & \multicolumn{9}{c}{0.6} \\
& $n$ & \multicolumn{3}{c|}{100} & \multicolumn{3}{c|}{200} & \multicolumn{3}{c}{400}\\
& \textrm{SNR} & $0.1$ & $0.3$ & $0.5$ & $0.1$ & $0.3$ & $0.5$  & $0.1$ & $0.3$ & $0.5$
\\ 
\hline 
\multirow{5}{*}{$\text{Power}_\text{strong}$}  		& LDPE $\lambda_{\text{min}}$	S-Z		&0.327 & 0.827 & 0.960 & 0.700 & 0.983 & 0.997 & 0.968 & 1.000 & 1.000 \\
  										& SSLasso $\lambda_{\text{min}}$ S-Z	&0.350 & 0.853 & 0.957 & 0.687 & 0.990  & 0.997 & 0.945 &  0.996 & 1.000 \\
										& dScore $\lambda_{\text{min}}$ S-Z		&0.297 & 0.787 & 0.937 & 0.697 & 0.987 & 0.993 & 0.968 & 0.996 & 1.000 \\
 										& nScore $\lambda_{\text{sup}}$ S-Z 	&0.420 & 0.870 & 0.957 & 0.720 & 0.987 & 1.000 & 0.947 & 1.000 & 1.000 \\
 										& nScore $\lambda_{\text{1SE}}$ S-Z	&0.350 & 0.800 & 0.927 & 0.717 & 0.980 & 1.000 & 0.968 & 1.000 & 1.000 \\
 										& nScore $\lambda_{\text{1SE}}$ T2.2 	&0.373 & 0.797 & 0.890 & 0.653 & 0.980 & 1.000 & 0.917 & 1.000 & 1.000 \\
\hline
\multirow{5}{*}{$\text{Power}_\text{weak}$}		& LDPE $\lambda_{\text{min}}$	S-Z		&0.043 & 0.049 & 0.046 & 0.041 & 0.077 & 0.063 & 0.053 & 0.066 & 0.083 \\
 										& SSLasso $\lambda_{\text{min}}$ S-Z	&0.061 & 0.054 & 0.070 & 0.053 & 0.086 & 0.083 & 0.067 & 0.099 & 0.105  \\
										& dScore $\lambda_{\text{min}}$ S-Z		&0.044 & 0.047 & 0.046 & 0.040 & 0.081 & 0.069 & 0.063 & 0.077 & 0.098 \\
 										& nScore $\lambda_{\text{sup}}$ S-Z 	&0.054 & 0.073 & 0.093 & 0.063 & 0.093 & 0.094 & 0.061 & 0.094 & 0.096 \\
 										& nScore $\lambda_{\text{1SE}}$ S-Z	&0.059 & 0.056 & 0.044 & 0.054 & 0.087 &  0.074 & 0.067 & 0.086 & 0.103 \\
 										& nScore $\lambda_{\text{1SE}}$ T2.2	&0.046 & 0.054 & 0.060 & 0.059 & 0.064 & 0.084 & 0.051 & 0.087 & 0.116 \\
\hline
\multirow{5}{*}{T1 Error} 						& LDPE $\lambda_{\text{min}}$	S-Z		&0.049 & 0.049 & 0.049 & 0.049 & 0.050 & 0.049 & 0.049 & 0.047 & 0.048 \\
 										& SSLasso $\lambda_{\text{min}}$ S-Z	&0.057 & 0.056 & 0.056 & 0.053 & 0.054 & 0.054 & 0.053 & 0.053 & 0.053 \\
										& dScore $\lambda_{\text{min}}$ S-Z		&0.033 & 0.039 & 0.036 & 0.031 & 0.033 & 0.033 & 0.032 & 0.030 & 0.031  \\
 										& nScore $\lambda_{\text{sup}}$ S-Z 	&0.063 & 0.079 & 0.089 & 0.062 & 0.077 & 0.075 & 0.060 & 0.062 & 0.048 \\
 										& nScore $\lambda_{\text{1SE}}$ S-Z	&0.057 & 0.051 & 0.047 & 0.056 & 0.049 & 0.039 & 0.055 & 0.045 & 0.046  \\
 										& nScore $\lambda_{\text{1SE}}$ T2.2	&0.051 & 0.051 & 0.051 & 0.049 & 0.048 & 0.046 & 0.047 & 0.045 & 0.045  \\
\hline \hline
\end{tabular}
\label{tab:scoreDSG35}
\end{table}

\subsection{Application to Glioblastoma Data}\label{sec:realdata}

We investigate a glioblastoma gene expression data set previously studied 
 in \citet{Horvathetal2006Data}. For each of $130$ patients, a survival outcome is available; we removed the twenty patients who were still alive at the end of the study. This resulted in a data set with $n=110$ observations. The gene expression measurements were normalized using the method of  \citet{Gautieretal2004affy}.
    We limited our analysis to $p=3600$ highly-connected genes  \citep{ZhangHorvath2005GBM, HorvathDong2008GBM}. The normalized data can be found at the website of Dr. Steve Horvath of UCLA Biostatistics.
     We log-transformed the survival response and centered it to have mean zero. Furthermore, we log-transformed the expression data, and then standardized each gene to have mean zero and standard deviation one
 across the $n=110$ observations.

  Our goal is to identify individual genes whose expression levels are associated with survival time, after adjusting for the other $3599$ genes in the data set. 
  With family-wise error rate (FWER) controlled at  level 0.1 using the Holm procedure \citep{Holm1979FWER}, 
 the na\"ive score test identifies three such genes: CKS2, H2AFZ, and RPA3.  \citet{Youetal2015CKS2} observed that CKS2 is highly expressed in glioma. \citet{Vardabassoetal2014hist} found  that histone genes, of which H2AFZ is one, are related to cancer progression. \citet{Jinetal2015RPA3} found that RPA3  is associated with glioma development. As a comparison, SSLasso finds two genes associated with patient survival: PPAP2C and RGS3. LDPE and dScore  
identify no genes at FWER of $0.1$.

\section{Discussion}\label{sec:disc}

In this paper, we examined  a  very na\"{i}ve two-step approach to high-dimensional inference: 
\begin{enumerate}
\item  Perform the lasso in order to select a small set of variables, $\hat\A_\lambda$. 
\item  Fit a least squares regression model using just the variables in $\hat\A_\lambda$, and make use of standard regression inference tools. Make no adjustment for the fact that $\hat\A_\lambda$ was selected based on the data.
\end{enumerate}
 It seems clear that this na\"{i}ve approach is problematic, since we have peeked at the data twice, but are not accounting for this double-peeking in our analysis.

In this paper, we have shown that under certain assumptions, $\hat\A_\lambda$ converges with high probability to a  deterministic set, $\A_\lambda$. A similar result for random design matrix is presented in \citet{Zhaoetal2019DCA}. This key insight allows us to establish that the confidence intervals resulting from the aforementioned na\"{i}ve two-step approach have  asymptotically correct coverage, in the sense of \eqref{eq:goalofposi}. This constitutes a theoretical justification for the recent simulation findings of  \citet{LeebetalPoSI2015}. 
  Furthermore, we used this key insight in order to establish that the   score test that results from the na\"{i}ve two-step approach  has asymptotically the same distribution as though the selected set of variables had been fixed in advance; thus, it can be used to test the null hypothesis $H_{0,j}^\ast:\beta^\ast_j=0$, $j=1,\dots,p$. 
  
  Our simulation results corroborate our theoretical findings. In fact, we find essentially no difference between the empirical performance of these na\"{i}ve proposals, and a host of other recent proposals in the literature for high-dimensional inference \citep{javanmard2013confidence, ZhangZhang2014LDPE, vandeGeeretal2014LDPE, lee2015exact, NingLiu2015decor}.

From a bird's-eye view, the recent literature on high-dimensional inference falls into two camps. The work of \citet{WassermanRoeder2009posi, Meinshausenetal2009posi, berk2013valid, lee2015exact, Tibshiranietal2016PoSI} focuses on  performing inference on the sub-model \eqref{eq:submodel}, whereas the work of \citet{JavanmardMontanari2013GLMSSLasso, javanmard2013confidence, JavanmardMontanari2014SDLTheory, ZhangZhang2014LDPE, vandeGeeretal2014LDPE, ZhaoShojaie2015Grace, NingLiu2015decor} focuses on testing hypotheses associated with \eqref{eq:linmodel}. In this paper, we have shown that the confidence intervals that result from the na\"{i}ve approach can be used to perform inference on the sub-model \eqref{eq:submodel}, whereas the score test that results from the na\"{i}ve approach can be used to test hypotheses associated with \eqref{eq:linmodel}.

 In the era of big data, simple analyses that are easy to apply and easy to understand are especially attractive to scientific investigators. Therefore, a careful investigation of such simple approaches is worthwhile, in order to determine which ones have the potential to yield accurate results, and which do not. 
 We do not advocate applying the na\"{i}ve two-step approach described above in most practical data analysis settings: we are confident that in practice, our intuition is correct, and this approach will perform poorly when the sample size is small or moderate, and/or the assumptions, which are unfortunately unverifiable, are not met. 
   However, in very large data settings, our results suggest that this na\"{i}ve approach may indeed be viable for high-dimensional inference, or at least warrants further investigation. 

When choosing among existing inference procedures based on lasso, the \emph{target of inference} should be taken into consideration. The target of inference can either be the population parameters, $\bbeta^\ast$ in \eqref{eq:linmodel}, or the parameters induced by the sub-model chosen by lasso, $\bbeta^{(\cM)}$ in \eqref{eq:betaM}. Sample-splitting \citep{WassermanRoeder2009posi, Meinshausenetal2009posi} and exact post selection \citep{lee2015exact, Tibshiranietal2016PoSI} methods provide valid inferences for $\bbeta^{(\cM)}$. The latter is a particularly appealing choice for inference on $\bbeta^{(\cM)}$, as it provides non-asymptotic confidence intervals under minimal assumptions. However, as we discussed in Section~\ref{sec:intro}, $\bbeta^{(\cM)}$ is, in general, different from $\bbeta^\ast_{\cM}$. A set of sufficient conditions for $\bbeta^{(\cM)} = \bbeta^\ast_{\cM}$ is the irrepresentable condition together with a beta-min condition. Unfortunately, these assumptions are unverifiable and may not hold in practice. 
Our theoretical analysis and empirical studies suggests that the na\"{i}ve two-step approach described above facilitates inference for $\bbeta^\ast_{\cM}$ under less stringent assumptions and without any conditioning or sample splitting. However, this method is also asymptotic and relies on unverifiable assumptions. 
Debiased lasso tests \citep{ZhangZhang2014LDPE, vandeGeeretal2014LDPE, JavanmardMontanari2013GLMSSLasso, javanmard2013confidence, NingLiu2015decor} provide asymptotically valid inference for entries of $\bbeta^\ast$, without requiring a beta-min or irrepresentable condition. However, they require more restrictive sparsity of $\bbeta^\ast$, as well as sparsity of the inverse covariance matrix of covariates, $\bSigma^{-1}$, which are also unverifiable. These limitations underscore the importance of recent efforts to relax these sparsity assumptions \citep[e.g.][]{zhu2018,wang2020}. 

We close with some suggestions for future research. One reviewer brought up an interesting point: methods with folded-concave penalties \citep[e.g.,][]{fan2001variable, Zhang2010MCP} require milder conditions to achieve variable selection consistency, i.e., $\Pr[\hat\A_\lambda=\A^\ast]\to1$, than the lasso. Inspired by this observation, we wonder whether the proposal of \citet{fan2001variable} and \citet{Zhang2010MCP} also require milder conditions to achieve $\Pr[\hat\A_\lambda=\A_\lambda]\to1$. If so, then we could replace the lasso with a folded-concave penalty in the variable selection step, and improve the robustness of the na\"ive approaches. We believe this could be a fruitful area of future research. In addition, extending the proposed theory and methods to generalized linear models  and M-estimators may also be promising areas for future research.


\appendix
\section{Proof of Proposition~\ref{THM:CONSISTENT}} \label{sec:pfconsistent}


We first state and prove Lemmas~\ref{basiclemma2}--\ref{T}, which are required  to prove Proposition~\ref{THM:CONSISTENT}. 

\begin{lemma}\label{basiclemma2}
Suppose ({\bf M1}) holds. Then, $\bbeta_\lambda$ and $\hat\bbeta_\lambda$ as defined in \eqref{eq:nllasso} and \eqref{eq:lasso}, respectively, are unique.
\end{lemma}
\begin{proof} 
The proof follows from Lemma 3 of \citet{Tibshirani2013} and is thus omitted.
\end{proof}


\begin{lemma}\label{basiclemma}
Suppose ({\bf M1}) holds. Then, 
\[
\frac{1}{n}\left\|\bX^\top\bepsilon\right\|_\infty=\mathcal{O}_p\left(\sqrt{\frac{\log(p)}{n}}\right).
\]
\end{lemma}
\begin{proof}
It is well understood that with an application of the union bound to the standard exponential inequality, ({\bf M1}) implies
\[
\frac{1}{n}\left\|\bX^\top\bepsilon\right\|_\infty\leq C\sqrt{\frac{2\log(p)}{n}}
\]
for some constant $C$ not depending on $p$ or $n$ \citep{Vershynin2012subGaussian}.
\end{proof}


\begin{lemma}\label{lem:qlambda}
Suppose ({\bf M1}), ({\bf M2}) and ({\bf E}) hold. Then $q_\lambda =\mathcal{O}(q^\ast)$. 
\end{lemma}
\begin{proof}
Recall that $b_{\lambda\min}\equiv\min_{j\in\A_\lambda}|\beta_{\lambda,j}|$. Hence
\begin{align}
b_{\lambda\min}^2 \left|\A_\lambda\backslash\A^\ast\right| \leq \left\|\bbeta_{\lambda, \A_\lambda\backslash\A^\ast}\right\|_2^2 &= \left\|\bbeta_{\lambda, \A_\lambda\backslash\A^\ast} - \bbeta^\ast_{\A_\lambda\backslash\A^\ast}\right\|_2^2 \leq \left\|\bbeta_{\lambda} - \bbeta^\ast\right\|_2^2.
\end{align}

With ({\bf E}), Lemma 2.1 in \citet{vandeGeerBuhlmann2009} shows that 
\begin{align}
\left\|\bbeta^\ast-\bbeta_\lambda\right\|_2\leq\frac{2\lambda\sqrt{2q^\ast}}{\phi^{\ast 2}},
\end{align}
i.e.,
\begin{align}
b_{\lambda\min}^2 \left|\A_\lambda\backslash\A^\ast\right| = \mathcal{O}\left(\lambda^2q^\ast\right).
\end{align}
Since $b_{\lambda\min} \succeq \lambda$, the above equation implies that $|\A_\lambda\backslash\A^\ast| = \mathcal{O}(q^\ast)$, which completes the proof.
\end{proof}

\begin{lemma}\label{mainlemmaa1}
Suppose ({\bf M1}) and ({\bf E}) hold. Then, the estimator $\hat \bbeta_\lambda$ defined in \eqref{eq:lasso} and its population version, $\bbeta_\lambda$, defined in \eqref{eq:nllasso} satisfy 
$$
\left\|\hat \bbeta_\lambda - \bbeta_{\lambda}\right\|_2 = \mathcal{O}_p\left(\sqrt{\frac{\log(p)}{n}}\right).
$$
\end{lemma}

\begin{proof}
Theorem 2.1 in \citet{vandeGeer2017} shows that under ({\bf M1}),  
\begin{align}\label{eq:vandegeer}
\left\|\bX\left(\hat\bbeta_\lambda-\bbeta_\lambda\right)\right\|_2\leq  \frac{1}{\lambda}\sqrt{\left(\bbeta^\ast-\bbeta_\lambda\right)^\top\hat\bSigma^2\left(\bbeta^\ast-\bbeta_\lambda\right)} \mathcal{O}_p(1) + \mathcal{O}_p\left(1\right).
\end{align}

Note there are three differences between \eqref{eq:vandegeer} and the formulation in Theorem 2.1 in \citet{vandeGeer2017}. 
First, \citet{vandeGeer2017} assumes $\|\bx_j\|_2^2=1, j=1,\dots,p$, whereas we assume $\|\bx_j\|_2^2=n$. This difference leads to the change of a $\sqrt{n}$ factor in the first term in \eqref{eq:vandegeer}. 
Second, Theorem 2.1 in \citet{vandeGeer2017} factors out a $\|\hat\bSigma\|_2$ from the square-root, whereas \eqref{eq:vandegeer} does not have this step, and consequently has a power of two on $\hat\bSigma$ within the square-root; see the proof of Theorem 15.2 in  \citet{vandeGeer2017} for more details.
Finally, although Theorem 2.1 in \citet{vandeGeer2017} assumes Gaussian random errors $\bepsilon$, it continues to hold for sub-Gaussian errors. This is because Lemma 15.5 in \citet{vandeGeer2017} can be proven with sub-Gaussian data as shown in \citet{Hsuetal2012}.

Let $q\equiv |\A_\lambda\backslash\A^\ast|$. Then, 
\begin{align}
\left(\bbeta^\ast-\bbeta_\lambda\right)^\top\hat\bSigma^2\left(\bbeta^\ast-\bbeta_\lambda\right)\leq& \left\|\bbeta^\ast-\bbeta_\lambda\right\|_2^2\sup_{\|\bb_{\A^{\ast c}}\|_0\leq q, \|\bb\|_2=1}\bb^\top\hat\bSigma^2\bb\nonumber \\
=& \phi^4(q) \left\|\bbeta^\ast-\bbeta_\lambda\right\|_2^2
\end{align}

With ({\bf E}), Lemma 2.1 in \citet{vandeGeerBuhlmann2009} shows that 
\begin{align}
\left\|\bbeta^\ast-\bbeta_\lambda\right\|_2\leq\frac{2\lambda\sqrt{2q^\ast}}{\phi^{\ast 2}}.
\end{align}
Thus, 
\begin{align}\label{ghjkl}
\frac{1}{\lambda}\sqrt{\left(\bbeta^\ast-\bbeta_\lambda\right)^\top\hat\bSigma^2\left(\bbeta^\ast-\bbeta_\lambda\right)} \leq \frac{1}{\lambda} \sqrt{\phi^4(q)\frac{8\lambda^2 q^\ast}{\phi^{\ast4}}} = \mathcal{O}\left(\frac{\phi^2(q)}{\phi^{\ast2}}\sqrt{q^\ast}\right).
\end{align}
Now, since $q\asymp q^\ast$ by Lemma~\ref{lem:qlambda}, $\phi^2(q)=\mathcal{O}\left(\sqrt{\log(p)/q^\ast}\right)$ by ({\bf E}). Thus, \eqref{ghjkl} implies that 
\begin{align}
\left\|\bX\left(\hat\bbeta_\lambda-\bbeta_\lambda\right)\right\|_2\leq\mathcal{O}\left(\frac{\phi^2(q)}{\phi^{\ast2}}\sqrt{q^\ast}\right)+\mathcal{O}_p\left(1\right)=\mathcal{O}_p\left(\sqrt{\log(p)}\right).
\end{align}

Let $\cI=\supp(\hat \bbeta_\lambda - \bbeta_{\lambda})$. Then, $| \cI | = \|\hat \bbeta_\lambda - \bbeta_{\lambda}\|_0 = \mathcal{O}_p(q^\ast)$. Moreover, $\|\hat \bbeta_{\lambda,\cI^c} - \bbeta_{\lambda,\cI^c}\|_1=0\leq\|\hat \bbeta_{\lambda,\cI} - \bbeta_{\lambda,\cI}\|_1$. Thus, by the restricted eigenvalue condition in ({\bf E}) that $\phi^{\ast 2}$ is a positive constant, 
\begin{align}
\left\|\hat \bbeta_{\lambda} - \bbeta_{\lambda}\right\|_2^2=\mathcal{O}_{p} \left(\frac{\phi^{\ast 2}}{n}\left\|\bX\left(\hat\bbeta_\lambda-\bbeta_\lambda\right)\right\|_2^2\right)= \mathcal{O}_p\left(\frac{\log(p)}{n}\right),
\end{align}
which completes the proof.
\end{proof}

\begin{lemma}\label{T}
Suppose ({\bf M1}), ({\bf M2}) and ({\bf T}) hold. Then
$$
\lim_{n\to\infty}\|\btau_{\lambda, \A_\lambda^c}\|_\infty \leq 1 - \delta,
$$
where 
\begin{align}
\label{eq:stationarya}
\lambda n\btau_{\lambda}&=\bX^\top\left(\bX\bbeta^\ast-\bX\bbeta_{\lambda}\right).
\end{align}
is the stationary condition of \eqref{eq:nllasso}.
\end{lemma}

\begin{proof}
Rearranging terms of \eqref{eq:stationarya}, we get 
\begin{align}\label{eq:tau}
\bbeta_{\lambda,\cA_\lambda} = \left(\bX^\top_{\cA_\lambda}\bX_{\cA_\lambda}\right)^{-1}\left(\bX^\top_{\cA_\lambda}\bX\bbeta^\ast-n\lambda\btau_{\lambda,\cA_\lambda}\right),
\end{align}
and
\begin{align}
\btau_\lambda &= \frac{1}{n\lambda}\left(\bX^\top\left(\bX\bbeta^\ast-\bX_{\cA_\lambda}\bbeta_{\lambda,\cA_\lambda}\right)\right) \nonumber \\
&= \frac{1}{n\lambda}\left(\bX^\top\left(\bX\bbeta^\ast-\bX_{\cA_\lambda}\left(\bX^\top_{\cA_\lambda}\bX_{\cA_\lambda}\right)^{-1}\left(\bX^\top_{\cA_\lambda}\bX\bbeta^\ast-n\lambda\btau_{\lambda,\cA_\lambda}\right)\right)\right) \nonumber \\
&=\frac{1}{n\lambda}\bX^\top\left(\bI-\bP^{(\A_\lambda)}\right)\bX\bbeta^\ast +\bX^\top\bX_{\cA_\lambda}\left(\bX^\top_{\cA_\lambda}\bX_{\cA_\lambda}\right)^{-1}\btau_{\lambda,\cA_\lambda} \nonumber \\
& = \frac{1}{n\lambda}\bX^\top\left(\bI-\bP^{(\A_\lambda)}\right)\bX_{\bA_\lambda^c}\bbeta^\ast_{\bA_\lambda^c} +\bX^\top\bX_{\cA_\lambda}\left(\bX^\top_{\cA_\lambda}\bX_{\cA_\lambda}\right)^{-1}\btau_{\lambda,\cA_\lambda},
\end{align}
where $\bP^{(\A_\lambda)}\equiv\bX_{\A_\lambda}(\bX_{\A_\lambda}^\top\bX_{\A_\lambda})^{-1}\bX^\top_{\A_\lambda}$ is the projection matrix onto the column space of $\bX_{\cA_\lambda}$. Therefore,
\begin{align}
\left\|\btau_{\lambda,\cA_\lambda^c}\right\|_\infty &\leq  \frac{1}{n\lambda}\left\|\bX_{\cA_\lambda^c}^\top\left(\bI-\bP^{(\A_\lambda)}\right)\bX_{\bA_\lambda^c}\bbeta^\ast_{\bA_\lambda^c}\right\|_\infty + \left\|\bX_{\cA_\lambda^c}^\top\bX_{\cA_\lambda}\left(\bX^\top_{\cA_\lambda}\bX_{\cA_\lambda}\right)^{-1}\btau_{\lambda,\cA_\lambda}\right\|_\infty \nonumber \\
&< \frac{1}{n\lambda} \max_{j\in\cA_\lambda^c} \left\{\left\|\bx_j\right\|_2\left\|\bX_{\bA_\lambda^c}\bbeta^\ast_{\bA_\lambda^c}\right\|_2\right\} + 1-\delta \nonumber \\
&=1-\delta + \mathcal{O}\left(\sqrt{\frac{\log(p)}{n}}\frac{1}{\lambda}\right)\label{mainT},
\end{align}
where the second inequality is due to the irrepresentable condition in ({\bf T}), and the last equality is due to the assumptions that $\bx_j^\top\bx_j=n$ in ({\bf M1}) and $\|\bX_{\bA_\lambda^c}\bbeta^\ast_{\bA_\lambda^c}\|_2=\mathcal{O}(\sqrt{\log(p)})$ in ({\bf M2}). Given that $\sqrt{\log(p)}/(\sqrt{n}\lambda)\to0$ by ({\bf M1}), we have 
\[
\lim_{n\to\infty}\|\btau_{\lambda,\cA_\lambda^c}\|_\infty \leq 1-\delta.
\]
\end{proof}

\begin{proof}[Proof of Proposition~\ref{THM:CONSISTENT}]
We first prove 
\begin{align}\label{side2}
\lim_{n\to\infty}\Pr\left[\A_\lambda\supseteq\hat\A_\lambda\right]=1. 
\end{align}

According to the stationary conditions of \eqref{eq:lasso},
\begin{align}
\label{eq:stationaryb}
\lambda n\hat\btau_{\lambda}&=\bX^\top\left(\by-\bX\hat\bbeta_{\lambda}\right).
\end{align}
Together with \eqref{eq:stationarya}, we get
\begin{align}\label{eq:tau}
\hat\btau_\lambda-\btau_\lambda = \frac{1}{n\lambda}\bX^\top\bX\left(\bbeta_{\lambda} -\hat\bbeta_{\lambda}\right) + \frac{1}{n\lambda}\bX^\top\beps.
\end{align}

We now bound both terms on the right hand side of \eqref{eq:tau}. By Lemma~\ref{basiclemma}, 
\[
\frac{\left\|\bX^\top\beps\right\|_\infty}{n\lambda}=\mathcal{O}_p\left(\frac{1}{\lambda}\sqrt{\frac{\log(p)}{n}}\right). 
\]
In addition, Lemma~\ref{mainlemmaa1} shows that 
\[
\left\|\bX\left(\hat\bbeta_\lambda-\bbeta_\lambda\right)\right\|_2=\mathcal{O}_p\left(\sqrt{\log(p)}\right).
\]
Therefore,
\begin{align*}
\frac{1}{n\lambda}\left\|\bX^\top\bX\left(\bbeta_{\lambda} -\hat\bbeta_{\lambda}\right)\right\|_\infty &\leq 
\frac{1}{n\lambda} \max_{j} \left\{\left\|\bx_j\right\|_2\left\|\bX\left(\bbeta_{\lambda} -\hat\bbeta_{\lambda}\right)\right\|_2\right\} \\
& =\mathcal{O}_p\left(\frac{1}{\lambda}\sqrt{\frac{\log(p)}{n}}\right).
\end{align*}
 
Therefore, it follows from \eqref{eq:tau} that $\|\btau_\lambda-\hat\btau_\lambda\|_\infty=\mathcal{O}_p\left(\sqrt{\log(p)/n}/\lambda\right)$. Given that $\sqrt{\log(p)/n}/\lambda\to0$ by ({\bf M1}), we have $\|\btau_\lambda-\hat\btau_\lambda\|_\infty=\smallO_p(1)$. Since $\lim_{n\to\infty}\|\btau_{\lambda, \A_\lambda^c}\|_\infty \leq 1 - \delta$ by Lemma~\ref{T}, we have $\lim_{n\to\infty}\Pr[\|\hat\btau_{\lambda, \A^c_\lambda}\|_\infty<1]=1$, and
\begin{align}
\lim_{n\to\infty}\Pr\left[\A_\lambda\supseteq\hat\A_\lambda\right]=1. 
\end{align}

To prove the other direction, we assume without loss of generality that $\A_\lambda\neq\emptyset$; otherwise the result holds trivially. Now, ({\bf M2}) and the fact that $\sqrt{\log(p)}/(\sqrt{n}\lambda)\to0$ by ({\bf M1}) imply that
\begin{align}\label{eq:cont1}
\sqrt{\frac{\log(p)}{n}}\frac{1}{b_{\lambda \min}}\to0.
\end{align}
But, by Lemma~\ref{mainlemmaa1}, $\|\hat \bbeta_\lambda - \bbeta_{\lambda}\|_2 = \mathcal{O}_p\left(\sqrt{\log(p)/n}\right)$. Thus, for any $\xi>0$, there exists a constant $C>0$, not depending on $n$, such that for $n$ sufficiently large, 
\begin{align}\label{eq:cont2}
\Pr\left[\left\|\hat \bbeta_\lambda - \bbeta_{\lambda}\right\|_\infty > C\sqrt{\frac{\log(p)}{n}}\,\right] <\xi.
\end{align}
Moreover, by \eqref{eq:cont1}, for $n$ sufficiently large, $b_{\lambda \min} > C\sqrt{\log(p)/n}$, where $b_{\lambda \min}\equiv\min_{j\in\A_\lambda}|\beta_{\lambda, j}|$. Combining \eqref{eq:cont1} and \eqref{eq:cont2}, for $n$ sufficiently large, whenever $|\beta_{\lambda, j}| > 0$, $|\beta_{\lambda, j}|  > C\sqrt{\log(p)/n}$ and hence $\Pr[|\hat\beta_{\lambda, j}|  > 0]>1-\xi$. Therefore
\begin{align}\label{side1}
\lim_{n\to\infty}\Pr\left[\A_\lambda\subseteq\hat\A_\lambda\right]=1,
\end{align}
which completes the proof.
\end{proof}


\section{Proof of Lemma~\ref{LEM:T}} \label{sec:pfT}

\begin{proof}[Proof of Lemma~\ref{LEM:T}]
We rewrite condition ({\bf T}) as 
\begin{align}
& \quad\left\|\bX_{\A_\lambda^c}^\top\bX_{\A_\lambda}\left(\bX^\top_{\A_\lambda}\bX_{\A_\lambda}\right)^{-1}\sign\left(\bbeta_{\lambda,\A_\lambda}\right)\right\|_\infty \nonumber \\
&= \left\|\bX_{\A_\lambda^c}^\top\bX_{\A_\lambda}\left(\bX^\top_{\A_\lambda}\bX_{\A_\lambda}\right)^{-1}\btau_{\lambda, \A_\lambda}\right\|_\infty \nonumber \\
& = \frac{1}{\lambda n}\left\|\bX_{\A_\lambda^c}^\top\bX_{\A_\lambda}\left(\bX^\top_{\A_\lambda}\bX_{\A_\lambda}\right)^{-1}\bX_{\A_\lambda}^\top\left(\bX\bbeta^\ast-\bX\bbeta_{\lambda}\right)\right\|_\infty \nonumber \\
& = \frac{1}{\lambda n}\left\|\bX_{\A_\lambda^c}^\top\bX_{\A_\lambda}\left(\bX^\top_{\A_\lambda}\bX_{\A_\lambda}\right)^{-1}\bX_{\A_\lambda}^\top\left(\bX_{\A_\lambda}\left(\bbeta^\ast_{\A_\lambda}-\bbeta_{\lambda,\A_\lambda}\right)+\bX_{\A_\lambda^c}\bbeta^\ast_{\A_\lambda^c}\right)\right\|_\infty \nonumber \\
&\leq \frac{1}{\lambda n}\left\|\bX_{\A_\lambda^c}^\top\bX_{\A_\lambda}\left(\bbeta^\ast_{\A_\lambda}-\bbeta_{\lambda,\A_\lambda}\right)\right\|_\infty + \frac{1}{\lambda n}\left\|\bX_{\A_\lambda^c}^\top \bP^{(\A_\lambda)}\bX_{\A_\lambda^c}\bbeta^\ast_{\A_\lambda^c}\right\|_\infty \nonumber.
\end{align}

With ({\bf E}), Lemma 2.1 in \citet{vandeGeerBuhlmann2009} shows that 
\begin{align*}
\left\|\bbeta^\ast-\bbeta_\lambda\right\|_2\leq\frac{2\lambda\sqrt{2q^\ast}}{\phi^{\ast 2}}.
\end{align*}
Based on \eqref{mainT}, ({\bf M2}) implies that 
\[
\frac{1}{n\lambda}\left\|\bX_{\cA_\lambda^c}^\top\bP^{(\A_\lambda)}\bX_{\bA_\lambda^c}\bbeta^\ast_{\bA_\lambda^c}\right\|_\infty =\mathcal{O}\left(\sqrt{\frac{\log(p)}{n}}\frac{1}{\lambda}\right).
\]

Thus,
\begin{align}
\Big\|\bX_{\A_\lambda^c}^\top\bX_{\A_\lambda}&\left(\bX^\top_{\A_\lambda}\bX_{\A_\lambda}\right)^{-1}\sign\left(\bbeta_{\lambda,\A_\lambda}\right)\Big\|_\infty \nonumber \\
&\leq \frac{1}{\lambda n}\left\|\bX_{\A_\lambda^c}^\top\bX_{\A_\lambda}\left(\bbeta^\ast_{\A_\lambda}-\bbeta_{\lambda,\A_\lambda}\right)\right\|_\infty + \frac{1}{\lambda n}\left\|\bX_{\A_\lambda^c}^\top \bP^{(\A_\lambda)}\bX_{\A_\lambda^c}\bbeta^\ast_{\A_\lambda^c}\right\|_\infty \nonumber \\
&\leq \frac{1}{\lambda n}\left\|\bX_{\A_\lambda^c}^\top\bX_{\A_\lambda}\right\|_\infty\left\|\bbeta^\ast_{\A_\lambda}-\bbeta_{\lambda,\A_\lambda}\right\|_\infty + \mathcal{O}\left(\sqrt{\frac{\log(p)}{n}}\frac{1}{\lambda}\right) \nonumber \\
&\leq  \frac{2\sqrt{2q^\ast}}{\phi^{\ast 2} n}\left\|\bX_{\A_\lambda^c}^\top\bX_{\A_\lambda}\right\|_\infty +  \mathcal{O}\left(\sqrt{\frac{\log(p)}{n}}\frac{1}{\lambda}\right).
\end{align}
Given that $\sqrt{\log(p)}/(\sqrt{n}\lambda)\to0$ by ({\bf M1}), condition ({\bf T}) is satisfied if 
\[
\frac{2\sqrt{2q^\ast}}{\phi^{\ast 2} n}\left\|\bX_{\A_\lambda^c}^\top\bX_{\A_\lambda}\right\|_\infty < 1-\delta,
\]
which completes the proof.
\end{proof}

\section{Proof of Theorem~\ref{THM:WALD2}} \label{sec:pfwald1}

\begin{proof}[Proof of Theorem~\ref{THM:WALD2}] By Proposition~\ref{THM:CONSISTENT}, $\Pr[\hat \A_\lambda= \A_\lambda]\to 1$. Therefore, with probability tending to one,
\begin{align}
\tilde\beta_{j}^{(\hat\A_\lambda)}\equiv\left[\left(\bX_{\hat\A_\lambda}^\top\bX_{\hat\A_\lambda}\right)^{-1}\bX_{\hat\A_\lambda}^\top\by\right]_j = \left[\left(\bX_{\A_\lambda}^\top\bX_{\A_\lambda}\right)^{-1}\bX_{\A_\lambda}^\top\by\right]_j. \label{waldolsequiv}
\end{align}
Thus,
\begin{align}
\tilde\beta_{j}^{(\hat\A_\lambda)}&=\left[\left(\bX_{\A_\lambda}^\top\bX_{\A_\lambda}\right)^{-1}\bX_{\A_\lambda}^\top\left(\bX\bbeta^\ast+\bepsilon\right)\right]_j \nonumber\\
& = \left[ \left(\bX_{\A_\lambda}^\top\bX_{\A_\lambda}\right)^{-1}\bX_{\A_\lambda}^\top\bepsilon\right]_j + \left[\left(\bX_{\A_\lambda}^\top\bX_{\A_\lambda}\right)^{-1}\bX_{\A_\lambda}^\top\bX\bbeta^\ast\right]_j . \label{waldolsstat}
\end{align}

We proceed to prove the asymptotic distribution of $[ (\bX_{\A_\lambda}^\top\bX_{\A_\lambda})^{-1}\bX_{\A_\lambda}^\top\bepsilon]_j$. Dividing it by its standard deviation, $\sigma_\bepsilon\sqrt{[(\bX_{\A_\lambda}^\top\bX_{\A_\lambda})^{-1}]_{(j, j)}}$, where $\sigma_\bepsilon$ is the error standard deviation, we get
\begin{align}
\frac{\left[\left(\bX_{\A_\lambda}^\top\bX_{\A_\lambda}\right)^{-1}\bX_{\A_\lambda}^\top\bepsilon\right]_j}{\sigma_\bepsilon\sqrt{\left[\left(\bX_{\A_\lambda}^\top\bX_{\A_\lambda}\right)^{-1}\right]_{(j, j)}}}=\frac{{\bf r}^w\beps}{\sigma_\beps\left\|{\bf r}^w\right\|_2},\label{asynormwald}
\end{align}
where $\br^w \equiv  \be^j(\bX^\top_{\A_\lambda}\bX_{\A_\lambda})^{-1}\bX^\top_{\A_\lambda} \in \mathbb{R}^n$, and $\be^j$ is the row vector of length $|\A_\lambda|$ with the entry that corresponds to $\beta^\ast_j$ equal to one, and zero otherwise. In order to use the Lindeberg-Feller Central Limit Theorem to prove the asymptotic normality of \eqref{asynormwald}, we need to show that the Lindeberg's condition holds, i.e.,
$$
\lim_{n\to\infty} \sum_{i=1}^n\E \left[\frac{ \left(r^w_{i} \epsilon_i\right)^2}{\sigma_\beps^2\left\|{\bf r}^w\right\|^2_2} \mathbf{1}\left\{ \frac{\left|r^w_{i} \epsilon_i\right|}{\sigma_\beps\left\|{\bf r}^w\right\|_2} >  \eta \right \}\right]= 0, \quad \forall \eta > 0.
$$
Given that $|r^w_{i}| \leq \|{\bf r}^w\|_\infty$, and that the $\epsilon_i$'s are identically distributed,
\begin{align*} 
0\leq\sum_{i=1}^n\E \left[\frac{ \left(r^w_{i} \epsilon_i\right)^2}{\sigma_\beps^2\left\|{\bf r}^w\right\|^2_2} \mathbf{1}\left\{ \frac{\left|r^w_{i} \epsilon_i\right|}{\sigma_\beps\left\|{\bf r}^w\right\|_2} >  \eta \right \}\right] &\leq \sum_{i=1}^n\E \left[\frac{ \left(r^w_{i} \epsilon_i\right)^2}{\sigma_\beps^2\left\|{\bf r}^w\right\|^2_2} \mathbf{1}\left\{ \frac{\left| \epsilon_i\right|\left\|{\bf r}^w\right\|_\infty }{\sigma_\beps\left\|{\bf r}^w\right\|_2 } >  \eta \right \} \right]\\
&=\sum_{i=1}^n \frac{r^{w2}_{i} }{\sigma_\beps^2\left\|{\bf r}^w \right\|_2^2} \E\left[ \epsilon_i^2\mathbf{1}\left\{ \frac{\left| \epsilon_i\right|\left\|{\bf r}^w\right\|_\infty }{\sigma_\beps\left\|{\bf r}^w\right\|_2 } >  \eta \right \} \right] \\
& = \frac{1}{\sigma_\beps^2} \E\left[ \epsilon_1^2 \mathbf{1}\left\{ \frac{\left| \epsilon_1\right|\left\|{\bf r}^w\right\|_\infty }{\sigma_\beps\left\|{\bf r}^w\right\|_2 } >  \eta \right \} \right].
\end{align*}
Since $\|{\bf r}^w\|_\infty/\|{\bf r}^w\|_2\rightarrow 0$ by Condition ({\bf W}), $\epsilon_1^2 \mathbf{1}\left\{ | \epsilon_1|\|{\bf r}^w\|_\infty/( \sigma_\beps\|{\bf r}^w\|_2)  >  \eta \right \} \rightarrow_p 0$. Thus, because $\epsilon_1^2 \geq \epsilon_1^2 \mathbf{1}\left\{ | \epsilon_1|\|{\bf r}^w\|_\infty/( \sigma_\beps\|{\bf r}^w\|_2)  >  \eta \right \}$ with probability one and $\E[\epsilon_1^2]=\sigma_\epsilon^2<\infty$, we use $\epsilon_1^2$ as the dominant random variable, and apply the Dominated Convergence Theorem, 
$$
\lim_{n\to\infty}\frac{1}{\sigma_\beps^2} \E\left[ \epsilon_1^2 \mathbf{1}\left\{ \frac{\left| \epsilon_1\right|\left\|{\bf r}^w\right\|_\infty }{\sigma_\beps\left\|{\bf r}^w\right\|_2 } >  \eta \right \} \right]= \frac{1}{\sigma_\beps^2}\E\left[\lim_{n\to\infty}\epsilon_1^2 \mathbf{1}\left\{ \frac{\left| \epsilon_1\right|\left\|{\bf r}^w\right\|_\infty }{\sigma_\beps\left\|{\bf r}^w\right\|_2 } >  \eta \right \} \right]=0,
$$
which gives the Lindeberg's condition.

Thus,
\begin{align}
\frac{\left[\left(\bX_{\A_\lambda}^\top\bX_{\A_\lambda}\right)^{-1}\bX_{\A_\lambda}^\top\bepsilon\right]_j}{\sigma_\bepsilon\sqrt{\left[\left(\bX_{\A_\lambda}^\top\bX_{\A_\lambda}\right)^{-1}\right]_{(j, j)}}}\rightarrow_d \mathcal{N}(0, 1).
\end{align}
Using, again, the fact that by Proposition~\ref{THM:CONSISTENT}, $\lim_{n\to\infty}\Pr\left[\A_\lambda=\hat\A_\lambda\right]=1$, we can write 
\begin{align*}
\frac{\tilde\beta_j^{(\hat\A_\lambda)}-\beta_j^{(\hat\A_\lambda)}}{\sigma_\bepsilon\sqrt{\left[\left(\bX_{\hat\A_\lambda}^\top\bX_{\hat\A_\lambda}\right)^{-1}\right]_{(j, j)}}}&\equiv\frac{\tilde\beta_j^{(\hat\A_\lambda)}-\left[\left(\bX_{\hat\A_\lambda}^\top\bX_{\hat\A_\lambda}\right)^{-1}\bX_{\hat\A_\lambda}^\top\bX\bbeta^\ast\right]_j}{\sigma_\bepsilon\sqrt{\left[\left(\bX_{\hat\A_\lambda}^\top\bX_{\hat\A_\lambda}\right)^{-1}\right]_{(j, j)}}}  \\
&\rightarrow_{p}\frac{\left[\left(\bX_{\A_\lambda}^\top\bX_{\A_\lambda}\right)^{-1}\bX_{\A_\lambda}^\top\bepsilon\right]_j}{\sigma_\bepsilon\sqrt{\left[\left(\bX_{\A_\lambda}^\top\bX_{\A_\lambda}\right)^{-1}\right]_{(j, j)}}} \rightarrow_d \mathcal{N}\left(0, 1\right).
\end{align*}
\end{proof}


\section{Proof of Theorem~\ref{THM:ERRORVAR}}\label{sec:pferrorvar}
\begin{proof}
Based on Proposition~\ref{THM:CONSISTENT} that $\Pr[\hat\A_\lambda=\A_\lambda]\to1$, we have
\begin{align}
\frac{1}{n-\hat q_\lambda}\left\|\by-\bX_{\hat\A_\lambda}\tilde\bbeta^{(\hat\A_\lambda)}\right\|_2^2&\to_p\frac{1}{n-q_\lambda}\left\|\by-\bX_{\A_\lambda}\tilde\bbeta^{(\A_\lambda)}\right\|_2^2,
\end{align}
where $\hat q_\lambda\equiv|\hat\A_\lambda|$, $q_\lambda\equiv|\A_\lambda|$. Denoting 
$\bP^{(\A_\lambda)}\equiv\bX_{\A_\lambda}(\bX_{\A_\lambda}^\top\bX_{\A_\lambda})^{-1}\bX^\top_{\A_\lambda}$, 
\begin{align}
&\quad\left\|\by-\bX_{\A_\lambda}\tilde\bbeta^{(\A_\lambda)}\right\|_2^2 \nonumber \\
&\equiv\left\|\by-\bX_{\A_\lambda}\left(\bX_{\A_\lambda}^\top\bX_{\A_\lambda}\right)^{-1}\bX_{\A_\lambda}\by\right\|_2^2 \nonumber \\
&=\by^\top\left(\bI-\bP^{(\A_\lambda)}\right)^2\by \nonumber \\
&=\by^\top\left(\bI-\bP^{(\A_\lambda)}\right)\by \nonumber \\
& = \left(\bX_{\A_\lambda}\bbeta^\ast_{\A_\lambda} +\bX_{\A_\lambda^c}\bbeta^\ast_{\A_\lambda^c}+\bepsilon\right)^\top\left(\bI-\bP^{(\A_\lambda)}\right)\left(\bX_{\A_\lambda}\bbeta^\ast_{\A_\lambda} +\bX_{\A_\lambda^c}\bbeta^\ast_{\A_\lambda^c}+\bepsilon\right) \nonumber \\
& = \left(\bX_{\A_\lambda^c}\bbeta^\ast_{\A_\lambda^c}+\bepsilon\right)^\top\left(\bI-\bP^{(\A_\lambda)}\right)\left(\bX_{\A_\lambda^c}\bbeta^\ast_{\A_\lambda^c}+\bepsilon\right) \nonumber \\
& = \left(\bX_{\A^\ast\backslash\A_\lambda}\bbeta^\ast_{\A^\ast\backslash\A_\lambda}+\bepsilon\right)^\top\left(\bI-\bP^{(\A_\lambda)}\right)\left(\bX_{\A^\ast\backslash\A_\lambda}\bbeta^\ast_{\A^\ast\backslash\A_\lambda}+\bepsilon\right) \label{eq:2.1.1}
\end{align}
where \eqref{eq:2.1.1} is based on the fact that $\bbeta^\ast_{-\A^\ast}\equiv\bzero$. To simplify notations, denote $\btheta\equiv\bX_{\A^\ast\backslash\A_\lambda}\bbeta^\ast_{\A^\ast\backslash\A_\lambda}$ and $\bQ\equiv\bI-\bP^{(\A_\lambda)}$. Expanding \eqref{eq:2.1.1},
\begin{align*}
\left(\btheta+\bepsilon\right)^\top\bQ\left(\btheta+\bepsilon\right) &=\btheta^\top\bQ\btheta+2\btheta^\top\bQ\bepsilon +\bepsilon^\top\bQ\bepsilon.
\end{align*}
Because $\bQ$ is an idempotent matrix, $\bQ$ is positive semidefinite, whose eigenvalues are all zeros and ones. Thus,
\begin{align*}
0\leq\btheta^\top\bQ\btheta\leq\phi^2_{\max}\left[\bQ\right]\left\|\btheta\right\|_2^2=\mathcal{O}\left(\log(p)\right),
\end{align*}
where the last equality is based on ({\bf M2}). Since $\log(p)/(n-q_\lambda)\to0$, 
\begin{align*}
0\leq\frac{1}{n-q_\lambda}\btheta^\top\bQ\btheta=\mathcal{O}\left(\frac{\log(p)}{n-q_\lambda}\right)=\smallO(1),
\end{align*}
which means that 
\begin{align}
\frac{1}{n-q_\lambda}\btheta^\top\bQ\btheta=\smallO(1). \label{eq:quad0}
\end{align}
Therefore, 
\begin{align}
\frac{1}{n-q_\lambda}\left\|\by-\bX_{\A_\lambda}\tilde\bbeta^{(\A_\lambda)}\right\|_2^2&=\frac{2}{n-q_\lambda}\btheta^\top\bQ\bepsilon +\frac{1}{n-q_\lambda}\bepsilon^\top\bQ\bepsilon +\smallO(1). \label{eq:reducedE}
\end{align}

We now derive the expected value and variance of $\|\by-\bX_{\A_\lambda}\tilde\bbeta^{(\A_\lambda)}\|_2^2/(n-q_\lambda)$. First, because $\E[\bepsilon]=\bzero$ and $\Cov[\bepsilon]=\sigma_\bepsilon^2\bI$, using expectation of quadratic forms,
\begin{align}
\quad\E\left[\frac{1}{n-q_\lambda}\left\|\by-\bX_{\A_\lambda}\tilde\bbeta^{(\A_\lambda)}\right\|_2^2\right] &= \E\left[\frac{2}{n-q_\lambda}\btheta^\top\bQ\bepsilon\right] +\E\left[\frac{1}{n-q_\lambda}\bepsilon^\top\bQ\bepsilon\right] +\smallO(1) \nonumber \\
&= \E\left[\frac{1}{n-q_\lambda}\bepsilon^\top\bQ\bepsilon\right] +\smallO(1) \nonumber \\
&=\frac{1}{n-q_\lambda}\sigma_\bepsilon^2\trace\left[\bQ\right]+\smallO(1).\label{eq:Equadform}
\end{align}
Because $\bQ$ is an idempotent matrix, we have $\trace[\bQ]=n-q_\lambda$. Thus
\begin{align}
\E\left[\frac{1}{n-q_\lambda}\left\|\by-\bX_{\A_\lambda}\tilde\bbeta^{(\A_\lambda)}\right\|_2^2\right]\to\sigma_\bepsilon^2. \label{eq:ERSS}
\end{align}

We now calculate the variance of $\|\by-\bX_{\A_\lambda}\tilde\bbeta^{(\A_\lambda)}\|_2^2/(n-q_\lambda)$. Since
\begin{align}
&\quad\Var\left[\frac{1}{n-q_\lambda}\left\|\by-\bX_{\A_\lambda}\tilde\bbeta^{(\A_\lambda)}\right\|_2^2\right] \nonumber \\
&=\frac{1}{\left(n-q_\lambda\right)^2}\left(\E\left[\left\|\by-\bX_{\A_\lambda}\tilde\bbeta^{(\A_\lambda)}\right\|_2^4\right]-\E\left[\left\|\by-\bX_{\A_\lambda}\tilde\bbeta^{(\A_\lambda)}\right\|_2^2\right]^2\right) \nonumber \\
&\to\frac{1}{\left(n-q_\lambda\right)^2}\E\left[\left\|\by-\bX_{\A_\lambda}\tilde\bbeta^{(\A_\lambda)}\right\|_2^4\right]-\sigma_\bepsilon^4, \label{eq:varRSS}
\end{align}
where the second term is derived in \eqref{eq:ERSS}. We now derive the expected value of $\|\by-\bX_{\A_\lambda}\tilde\bbeta^{(\A_\lambda)}\|_2^4/(n-q_\lambda)^2$. Based on \eqref{eq:reducedE},
\begin{align}
\frac{1}{\left(n-q_\lambda\right)^2}\E\left[\left\|\by-\bX_{\A_\lambda}\tilde\bbeta^{(\A_\lambda)}\right\|_2^4\right] &= \smallO(1)\cdot\frac{1}{n-q_\lambda}\E\left[2\btheta^\top\bQ\bepsilon +\bepsilon^\top\bQ\bepsilon\right] \nonumber \\
&+ \frac{4}{\left(n-q_\lambda\right)^2}\E\left[\btheta^\top\bQ\bepsilon\bepsilon^\top\bQ\btheta\right] \nonumber \\
&+ \frac{2}{\left(n-q_\lambda\right)^2}\E\left[\btheta^\top\bQ\bepsilon\bepsilon^\top\bQ\bepsilon\right] \nonumber \\
& + \frac{1}{\left(n-q_\lambda\right)^2}\E\left[\bepsilon^\top\bQ\bepsilon\bepsilon^\top\bQ\bepsilon\right] + \smallO(1).
\end{align}
We consider each term in the above formulation separately. First,
\begin{align}
\frac{1}{n-q_\lambda}\E\left[2\btheta^\top\bQ\bepsilon +\bepsilon^\top\bQ\bepsilon\right] =\frac{1}{n-q_\lambda}\E\left[\bepsilon^\top\bQ\bepsilon\right] = \sigma_\bepsilon^2, \label{eq:term1}
\end{align}
where the last equality is based on \eqref{eq:Equadform} and \eqref{eq:ERSS}. For the second term, based on \eqref{eq:quad0},
\begin{align}
\frac{4}{\left(n-q_\lambda\right)^2}\E\left[\btheta^\top\bQ\bepsilon\bepsilon^\top\bQ\btheta\right]&= \frac{4}{\left(n-q_\lambda\right)^2}\btheta^\top\bQ\E\left[\bepsilon\bepsilon^\top\right]\bQ\btheta \nonumber \\
&=\frac{4\sigma_\bepsilon^2}{\left(n-q_\lambda\right)^2}\btheta^\top\bQ\bQ\btheta=\frac{4\sigma_\bepsilon^2}{\left(n-q_\lambda\right)^2}\btheta^\top\bQ\btheta=\smallO\left(\frac{1}{n-q_\lambda}\right).\label{eq:term2}
\end{align}
For the third term,
\begin{align}
&\quad\frac{2}{\left(n-q_\lambda\right)^2}\E\left[\btheta^\top\bQ\bepsilon\bepsilon^\top\bQ\bepsilon\right] \nonumber \\
&=\frac{2}{\left(n-q_\lambda\right)^2}\E\left[\btheta^\top\bQ\bepsilon\right]\E\left[\bepsilon^\top\bQ\bepsilon\right] + \Cor\left[\btheta^\top\bQ\bepsilon, \bepsilon^\top\bQ\bepsilon\right]\sqrt{\frac{2\Var\left[\btheta^\top\bQ\bepsilon\right]}{\left(n-q_\lambda\right)^2}\frac{2\Var\left[\bepsilon^\top\bQ\bepsilon\right]}{\left(n-q_\lambda\right)^2}} \nonumber \\
&=\Cor\left[\btheta^\top\bQ\bepsilon, \bepsilon^\top\bQ\bepsilon\right]\sqrt{\frac{2\Var\left[\btheta^\top\bQ\bepsilon\right]}{\left(n-q_\lambda\right)^2}\frac{2\Var\left[\bepsilon^\top\bQ\bepsilon\right]}{\left(n-q_\lambda\right)^2}}, \nonumber 
\end{align}
where, based on \eqref{eq:quad0},
\begin{align}
\frac{2}{\left(n-q_\lambda\right)^2}\Var\left[\btheta^\top\bQ\bepsilon\right] &=\frac{2}{\left(n-q_\lambda\right)^2}\left(\E\left[\btheta^\top\bQ\bepsilon\bepsilon^\top\bQ\btheta\right] - \E\left[\btheta^\top\bQ\bepsilon\right]^2\right) \nonumber \\
&=\frac{2\sigma_\bepsilon^2}{\left(n-q_\lambda\right)^2}\btheta^\top\bQ\btheta =\smallO\left(\frac{1}{n-q_\lambda}\right),\label{eq:term3a}
\end{align}
and
\begin{align}
\frac{2}{\left(n-q_\lambda\right)^2}\Var\left[\bepsilon^\top\bQ\bepsilon\right] &=\frac{2}{\left(n-q_\lambda\right)^2}\left(\E\left[\bepsilon^\top\bQ\bepsilon\bepsilon^\top\bQ\bepsilon\right] - \E\left[\bepsilon^\top\bQ\bepsilon\right]^2\right)  \nonumber \\
& = \frac{2}{\left(n-q_\lambda\right)^2} \E\left[\bepsilon^\top\bQ\bepsilon\bepsilon^\top\bQ\bepsilon\right] - 2\sigma_\bepsilon^4. \label{eq:term3b}
\end{align}
The last equality is based on \eqref{eq:Equadform} and \eqref{eq:ERSS}. Since $-1\leq\Cor[\btheta^\top\bQ\bepsilon, \bepsilon^\top\bQ\bepsilon]\leq1$, based on \eqref{eq:term3a} and \eqref{eq:term3b}, we have 
\begin{align}
-\sqrt{\smallO\left(\frac{1}{\left(n-q_\lambda\right)^3}\right)\E\left[\bepsilon^\top\bQ\bepsilon\bepsilon^\top\bQ\bepsilon\right]-\smallO\left(\frac{1}{n-q_\lambda}\right)} &\leq \frac{2}{\left(n-q_\lambda\right)^2}\E\left[\btheta^\top\bQ\bepsilon\bepsilon^\top\bQ\bepsilon\right] \nonumber \\
&\leq\sqrt{\smallO\left(\frac{1}{\left(n-q_\lambda\right)^3}\right)\E\left[\bepsilon^\top\bQ\bepsilon\bepsilon^\top\bQ\bepsilon\right]-\smallO\left(\frac{1}{n-q_\lambda}\right)}. \nonumber
\end{align}
Therefore, collecting \eqref{eq:term1}, \eqref{eq:term2}, \eqref{eq:term3a} and \eqref{eq:term3b}, we have
\begin{align}
0&\leq\frac{1}{\left(n-q_\lambda\right)^2}\E\left[\left\|\by-\bX_{\A_\lambda}\tilde\bbeta^{(\A_\lambda)}\right\|_2^4\right]\nonumber \\
&\leq\smallO\left(\frac{1}{\left(n-q_\lambda\right)^{3/2}}\right)\sqrt{\E\left[\bepsilon^\top\bQ\bepsilon\bepsilon^\top\bQ\bepsilon\right]} +\frac{1}{\left(n-q_\lambda\right)^2}\E\left[\bepsilon^\top\bQ\bepsilon\bepsilon^\top\bQ\bepsilon\right] + \smallO(1). \label{eq:collect}
\end{align}

But,
\begin{align}
\E\left[\bepsilon^\top\bQ\bepsilon\bepsilon^\top\bQ\bepsilon\right] &=\E\left[\sum_{i, j, l, k}^n\epsilon_i\epsilon_j\epsilon_l\epsilon_k Q_{(i,j)}Q_{(l,k)}\right] \nonumber \\
&=\E\left[\sum_{i=1}^n\epsilon_i^4Q_{(i, i)}^2\right]+\E\left[\sum_{i=l\neq j=k}^n\epsilon_i^2\epsilon_j^2Q_{(i, j)}Q_{(l,k)}\right] \nonumber \\
&+\E\left[\sum_{i=k\neq j=l}^n\epsilon_i^2\epsilon_j^2Q_{(i, j)}Q_{(l,k)}\right]+\E\left[\sum_{i=j\neq l=k}^n\epsilon_i^2\epsilon_l^2Q_{(i, j)}Q_{(l,k)}\right].
\end{align}
Because $\bepsilon$ is independent and identically distributed with $\E[\epsilon_1^2]=\sigma_\bepsilon^2$,
\begin{align}
\E\left[\bepsilon^\top\bQ\bepsilon\bepsilon^\top\bQ\bepsilon\right] &=\E\left[\epsilon_1^4\right]\sum_{i=1}^nQ_{(i, i)}^2+\E\left[\epsilon_1^2\right]^2\sum_{i=l\neq j=k}^nQ_{(i, j)}Q_{(l,k)} \nonumber \\
&+\E\left[\epsilon_1^2\right]^2\sum_{i=k\neq j=l}^nQ_{(i, j)}Q_{(l,k)}+\E\left[\epsilon_1^2\right]^2\sum_{i=j\neq l=k}^nQ_{(i, j)}Q_{(l,k)} \nonumber \\
&=\E\left[\epsilon_1^4\right]\sum_{i=1}^nQ_{(i, i)}^2+\sigma_\bepsilon^4\sum_{i\neq j}^nQ_{(i, j)}^2 + \sigma_\bepsilon^4\sum_{i\neq j}^nQ_{(i, j)}^2 + \sigma_\bepsilon^4\sum_{i\neq j}^nQ_{(i, i)}Q_{(j, j)} \nonumber \\ 
&=\E\left[\epsilon_1^4\right]\sum_{i=1}^nQ_{(i, i)}^2+2\sigma_\bepsilon^4\sum_{i\neq j}^nQ_{(i, j)}^2 + \sigma_\bepsilon^4\sum_{i\neq j}^nQ_{(i, i)}Q_{(j, j)}.
\end{align}
Moreover, because $(\sum_{i=1}^nQ_{(i,i)})^2=\sum_{i\neq j}^nQ_{(i, i)}Q_{(j, j)}+ \sum_{i=1}^nQ_{(i, i)}^2$,
\begin{align}
\E\left[\bepsilon^\top\bQ\bepsilon\bepsilon^\top\bQ\bepsilon\right] &=\E\left[\epsilon_1^4\right]\sum_{i=1}^nQ_{(i, i)}^2+2\sigma_\bepsilon^4\sum_{i\neq j}^nQ_{(i, j)}^2 + \sigma_\bepsilon^4\left(\sum_{i=1}^nQ_{(i, i)}\right)^2-\sigma_\bepsilon^4\sum_{i=1}^nQ_{(i, i)}^2 \nonumber \\
&=\E\left[\epsilon_1^4\right]\sum_{i=1}^nQ_{(i, i)}^2+\sigma_\bepsilon^4\sum_{i\neq j}^nQ_{(i, j)}^2 + \sigma_\bepsilon^4\left(\sum_{i=1}^nQ_{(i, i)}\right)^2 \nonumber \\
&\leq\E\left[\epsilon_1^4\right]\left(\sum_{i=1}^nQ_{(i, i)}^2+\sum_{i\neq j}^nQ_{(i, j)}^2\right) + \sigma_\bepsilon^4\left(\sum_{i=1}^nQ_{(i, i)}\right)^2 \nonumber \\
&= \E\left[\epsilon_1^4\right]\sum_{i,j}^nQ_{(i,j)}^2+\sigma_\bepsilon^4\trace\left[\bQ\right]^2,
\end{align}
where the inequality is based on Jensen's inequality that $\E[\epsilon_1^2]^2\leq\E[\epsilon_1^4]$. Because $\sum_{i,j}^nQ_{(i,j)}^2=\trace[\bQ^2]=\trace[\bQ]=n-q_\lambda$, we have 
\begin{align}
\E\left[\bepsilon^\top\bQ\bepsilon\bepsilon^\top\bQ\bepsilon\right] &= \left(n-q_\lambda\right)\E\left[\epsilon_1^4\right]+\left(n-q_\lambda\right)^2\sigma_\bepsilon^4.
\end{align}
Since $\epsilon_1$ has sub-Gaussian tails, we have $\E[\epsilon_1^4]=\mathcal{O}(1)$ \citep[see, e.g., Lemma 5.5 in][]{Vershynin2012subGaussian}. Thus, based on \eqref{eq:collect},
\begin{align}
0&\leq\frac{1}{\left(n-q_\lambda\right)^2}\E\left[\left\|\by-\bX_{\A_\lambda}\tilde\bbeta^{(\A_\lambda)}\right\|_2^4\right]\leq\sigma_\bepsilon^4+\smallO(1),
\end{align}
and hence, based on \eqref{eq:varRSS},
\begin{align}
\Var\left[\frac{1}{n-q_\lambda}\left\|\by-\bX_{\A_\lambda}\tilde\bbeta^{(\A_\lambda)}\right\|_2^2\right]=\smallO(1).
\end{align}
Finally, applying Chebyshev's inequality, we obtain
\begin{align}
\frac{1}{n-\hat q_\lambda}\left\|\by-\bX_{\hat\A_\lambda}\tilde\bbeta^{(\hat\A_\lambda)}\right\|_2^2&\to_p\frac{1}{n-q_\lambda}\left\|\by-\bX_{\A_\lambda}\tilde\bbeta^{(\A_\lambda)}\right\|_2^2 \nonumber \\
&\to_p\E\left[\frac{1}{n-q_\lambda}\left\|\by-\bX_{\A_\lambda}\tilde\bbeta^{(\A_\lambda)}\right\|_2^2\right]=\sigma_\bepsilon^2.
\end{align}
\end{proof}


\section{Proof of Theorem~\ref{THM:SCORE}}\label{sec:pfscore}

\begin{proof}[Proof of Theorem~\ref{THM:SCORE}]
By Proposition~\ref{THM:CONSISTENT}, $\Pr\left[\hat \A_\lambda= \A_\lambda\right]\to 1$. Therefore, we also have $\Pr\left[\big(\hat \A_\lambda\backslash\{j\}\big)= \big(\A_\lambda\backslash\{j\}\big)\right]\to 1$, and with probability tending to one,
\begin{align}
S^j\equiv\bx_j^\top\left({\bf I}_n-{\bf P}^{(\hat \A_\lambda\backslash\{j\})}\right)\by &=\bx_j^\top\left({\bf I}_n-{\bf P}^{( \A_\lambda\backslash\{j\})}\right)\by \nonumber \\
& =  \bx_j^\top\left({\bf I}_n-{\bf P}^{(\A_\lambda\backslash\{j\})}\right)\left(\bX_{\A_\lambda}\bbeta^\ast_{\A_\lambda}+ \bX_{\A_\lambda^c}\bbeta^\ast_{\A_\lambda^c}+\bepsilon\right)\label{teststatahat},
\end{align}
where ${\bf P}^{(\A_\lambda\backslash\{j\})}\equiv\bX_{\A_\lambda\backslash\{j\}}(\bX_{\A_\lambda\backslash\{j\}}^\top\bX_{\A_\lambda\backslash\{j\}})^{-1}\bX^\top_{\A_\lambda\backslash\{j\}}$.

Thus, under the null hypothesis $H_{0,j}^\ast:\beta^\ast_j=0$, \eqref{teststatahat} is equal to
\begin{align}
&\quad \bx_j^\top\left({\bf I}_n-{\bf P}^{(\A_\lambda\backslash\{j\})}\right) \left(\bX_{\A_\lambda\backslash\{j\}}\bbeta^\ast_{\A_\lambda\backslash\{j\}}+ \bX_{\A_\lambda^c\backslash\{j\}}\bbeta^\ast_{\A_\lambda^c\backslash\{j\}}+\bepsilon\right) \nonumber\\
&=\bx_j^\top\left(\bI_n-\bP^{(\A_\lambda\backslash\{j\})}\right)\bepsilon + \bx_j^\top\left(\bI_n-\bP^{(\A_\lambda\backslash\{j\})}\right)\bX_{\A_\lambda^c\backslash\{j\}}\bbeta^\ast_{\A_\lambda^c\backslash\{j\}}. \label{teststatahat2}
\end{align}
The equality holds because $({\bf I}_n-{\bf P}^{(\A_\lambda\backslash\{j\})})\bX_{\A_\lambda\backslash\{j\}}\bbeta^\ast_{\A_\lambda\backslash\{j\}}=\bzero$.

We first show the asymptotic distribution of $\bx_j^\top(\bI_n-\bP^{(\A_\lambda\backslash\{j\})})\bepsilon$. Dividing it by its standard deviation $\sigma_\beps\sqrt{\bx_j^\top({\bf I}_n - {\bf P}^{(\A_\lambda\backslash\{j\})})\bx_j}$, where $\sigma_\bepsilon$ is the error standard deviation,
\begin{align}
\frac{\bx_j^\top\left({\bf I}_n - {\bf P}^{(\A_\lambda\backslash\{j\})}\right)\bepsilon}{\sigma_\bepsilon\sqrt{\bx_j^\top\left({\bf I}_n - {\bf P}^{(\A_\lambda\backslash\{j\})}\right)\bx_j}} &= \frac{{\bf r}^{s\top}\beps}{\sigma_\beps\left\|{\bf r}^s\right\|_2}, \label{asynorm}
\end{align}
where  ${\bf r}^{s\top} \equiv \bx_j^\top({\bf I}_n - {\bf P}^{(\A_\lambda\backslash\{j\})})$. 
Now, we use the Lindeberg-Feller Central Limit Theorem to prove the asymptotic normality of \eqref{asynorm}. Similar to the proof of Theorem~\ref{THM:WALD2}, we need to prove that the Lindeberg's condition holds, i.e.,
$$
\lim_{n\to\infty} \sum_{i=1}^n\E \left[\frac{ \left(r^s_{i} \epsilon_i\right)^2}{\sigma_\beps^2\left\|{\bf r}^s\right\|^2_2} \mathbf{1}\left\{ \frac{\left|r^s_{i} \epsilon_i\right|}{\sigma_\beps\left\|{\bf r}^s\right\|_2} >  \eta \right \}\right]= 0, \quad \forall \eta > 0.
$$
Given that $|r^s_{i}| \leq \|{\bf r}^s\|_\infty$, and that the $\epsilon_i$'s are identically distributed,
\begin{align*} 
0\leq\sum_{i=1}^n\E \left[\frac{ \left(r^s_{i} \epsilon_i\right)^2}{\sigma_\beps^2\left\|{\bf r}^s\right\|^2_2} \mathbf{1}\left\{ \frac{\left|r^s_{i} \epsilon_i\right|}{\sigma_\beps\left\|{\bf r}^s\right\|_2} >  \eta \right \}\right]&\leq \frac{1}{\sigma_\beps^2} \E\left[ \epsilon_1^2 \mathbf{1}\left\{ \frac{\left| \epsilon_1\right|\left\|{\bf r}^s\right\|_\infty }{\sigma_\beps\left\|{\bf r}^s\right\|_2 } >  \eta \right \} \right].
\end{align*}
Since $\|{\bf r}^s\|_\infty/\|{\bf r}^s\|_2\rightarrow 0$ by Condition ({\bf S}), $\epsilon_1^2 \mathbf{1}\left\{ | \epsilon_1|\|{\bf r}^s\|_\infty/( \sigma_\beps\|{\bf r}^s\|_2)  >  \eta \right \} \rightarrow_p 0$. Thus, because $\epsilon_1^2 \geq \epsilon_1^2 \mathbf{1}\left\{ | \epsilon_1|\|{\bf r}^s\|_\infty/( \sigma_\beps\|{\bf r}^s\|_2)  >  \eta \right \}$ with probability one and $\E[\epsilon_1^2]=\sigma_\epsilon^2<\infty$, we use $\epsilon_1^2$ as the dominant random variable, and apply the Dominated Convergence Theorem,  
$$
\lim_{n\to\infty}\frac{1}{\sigma_\beps^2} \E\left[ \epsilon_1^2 \mathbf{1}\left\{ \frac{\left| \epsilon_1\right|\left\|{\bf r}^s\right\|_\infty }{\sigma_\beps\left\|{\bf r}^s\right\|_2 } >  \eta \right \} \right]= 0,
$$
which in turn gives the Lindeberg's condition. 

Thus,
\begin{align}
\frac{\bx_j^\top\left({\bf I}_n - {\bf P}^{(\A_\lambda\backslash\{j\})}\right)\bepsilon}{\sigma_\bepsilon\sqrt{\bx_j^\top\left({\bf I}_n - {\bf P}^{(\A_\lambda\backslash\{j\})}\right)\bx_j}}\rightarrow_d \mathcal{N}(0, 1),
\end{align}

We now prove the asymptotic unbiasedness of the na\"ive score test on $\bbeta^\ast$. Dividing the second term in \eqref{teststatahat2} by $\sigma_\beps\sqrt{\bx_j^\top({\bf I}_n - {\bf P}^{(\A_\lambda\backslash\{j\})})\bx_j}$, we get
\begin{align}
\left|\frac{\bx_j^\top\left(\bI_n-\bP^{(\A_\lambda\backslash\{j\})}\right)\bX_{\A_\lambda^c\backslash\{j\}}\bbeta^\ast_{\A_\lambda^c\backslash\{j\}}}{\sigma_\bepsilon\sqrt{\bx_j^\top\left({\bf I}_n - {\bf P}^{(\A_\lambda\backslash\{j\})}\right)\bx_j}}\right|&=\left|\frac{{\bf r}^s\bX_{\A_\lambda^c\backslash\{j\}}\bbeta^\ast_{\A_\lambda^c\backslash\{j\}}}{\sigma_\beps\left\|{\bf r}^s\right\|_2}\right| \nonumber \\
&\leq\frac{\left\|{\bf r}^s\right\|_2}{\left\|{\bf r}^s\right\|_2}\frac{\left\|\bX_{\A_\lambda^c\backslash\{j\}}\bbeta^\ast_{\A_\lambda^c\backslash\{j\}}\right\|_2}{\sigma_\bepsilon} \nonumber \\
&= \frac{\left\|\bX_{\A_\lambda^c\backslash\{j\}}\bbeta^\ast_{\A_\lambda^c\backslash\{j\}}\right\|_2}{\sigma_\bepsilon}.
\end{align}
By ({\bf M2$^\ast$}),
\begin{align*}
\left\|\bX_{\A_\lambda^c\backslash\{j\}}\bbeta^\ast_{\A_\lambda^c\backslash\{j\}}\right\|_2 = \left\|\bX_{(\A^\ast\backslash\A_\lambda)\backslash\{j\}}\bbeta^\ast_{(\A^\ast\backslash\A_\lambda)\backslash\{j\}}\right\|_2=\left\|\bX_{\A^\ast\backslash\A_\lambda}\bbeta^\ast_{\A^\ast\backslash\A_\lambda}\right\|_2=\smallO\left(1\right),
\end{align*}
where the second equality holds under $H_{0,j}:\beta^\ast_j=0$.

Using, again, the fact that by Proposition~\ref{THM:CONSISTENT}, $\lim_{n\to\infty}\Pr\left[\A_\lambda=\hat\A_\lambda\right]=1$, we get
\begin{align}
\frac{\bx_j^\top\left({\bf I}_n-{\bf P}^{(\hat \A_\lambda\backslash\{j\})}\right)\by}{\sigma_\bepsilon\sqrt{\bx_j^\top\left({\bf I}_n - {\bf P}^{(\hat\A_\lambda\backslash\{j\})}\right)\bx_j}}\rightarrow_d \mathcal{N}(0, 1).
\end{align}
\end{proof}

\section*{Acknowledgements}
We thank the authors of \citet{javanmard2013confidence, NingLiu2015decor} for providing code for their proposals. 
We are grateful to Joshua Loftus, Jonathan Taylor, Robert Tibshirani, and Ryan Tibshirani for helpful responses to our inquiries.


\label{lastpage}
\bibliographystyle{apalike}
\bibliography{LassoTests-rev4-v2}

\end{document}